\RequirePackage{fix-cm}
\documentclass[smallextended]{svjour3}       % onecolumn (second format)
\smartqed  % flush right qed marks, e.g. at end of proof

\usepackage{graphicx}
\usepackage{amssymb}
\usepackage{amsmath}
\usepackage{wasysym}
\usepackage{hyperref}
\usepackage{mathtools}
\usepackage{comment}

\newcommand{\Rn}[1][n]{\mathbb{R}^{#1}}
\newcommand{\R}{\mathbb{R}}
\newcommand{\cK}{\overline{\mathcal{K}}}
\newcommand{\K}{\mathcal{K}}
\newcommand{\Kmin}{\mathcal{K}_{\min}}
\newcommand{\oneC}{C_1^n}
\newcommand{\oneCv}{C_1^n(v)}

\begin{document}

\title{1-convex extensions of incomplete cooperative games and the average value
\thanks{Both authors were supported by SVV--2020--260578 and by the Charles University Grant Agency (GAUK 341721).}
}

\author{Jan Bok \and Martin \v{C}ern\'{y}
}

\institute{Jan Bok \at
              Computer Science Institute, Faculty of Mathematics and Physics, Charles University, Prague, Czech Republic \\
              ORCID: 0000-0002-7973-1361\\
              \email{bok@iuuk.mff.cuni.cz}             \\
              This author is the corresponding author.
           \and
            Martin \v{C}ern\'{y} \at
              Department of Applied Mathematics, Faculty of Mathematics and Physics, Charles University, Prague, Czech Republic \\
              ORCID: 0000-0002-0619-4737\\
              \email{cerny@kam.mff.cuni.cz}
}

\date{Received: date / Accepted: date}

\maketitle

\begin{abstract}
The model of incomplete cooperative games incorporates uncertainty into the classical model of cooperative games by considering a partial characteristic function. Thus the values for some of the coalitions are not known. The main focus of this paper is the class of 1-convex cooperative games under this framework.

We are interested in two heavily intertwined questions. First, given an incomplete game, in which ways can we fill in the missing values to obtain a classical 1-convex game? Such complete games are called \emph{1-convex extensions}. For the class of minimal incomplete games (in which precisely the values of singletons and grand coalitions are known), we provide an answer in terms of a description of the set of 1-convex extensions. The description employs extreme points and extreme rays of the set. We also provide bounds on sets of 1-convex extensions for such games.

Second, how to determine in a rational, fair, and efficient way the payoffs of players based only on the known values of coalitions? Based on the description of the set of 1-convex extensions, we introduce generalisations of three solution concepts (values) for complete games, namely the $\tau$-value, the Shapley value and the nucleolus. We consider two variants where we compute the centre of gravity of either extreme games or of a combination of extreme games and extreme rays. We show that all of the generalised values coincide for minimal incomplete games which allows to introduce the \emph{average value}. For this value, we provide three different axiomatisations based on axiomatic characterisations of the $\tau$-value and the Shapley value for classical cooperative games.

Finally, we turn our attention to \emph{incomplete games with defined upper vector}, asking the same questions and this time arriving to different conclusions. This provides a benchmark to test our tools and knowledge of the average value. This part highlights the importance and also the difficulty of considering more general classes of incomplete games.

\end{abstract}

\tableofcontents

\section{Introduction}

In the theory of (cooperative) games uncertainty is a long studied and a very important problem. The reasons are both practical and theoretical. Regarding applications, inaccuracy in data is relatively common in real-world situations. The sources of such inaccuracies can be for example the lack of knowledge on the behaviour of other agents, corrupted data, signal noise or events such as voting or auctions for which we do not know full information or it is very expensive to obtain.

Regarding the theoretical aspects, there is clearly no single ideal way to tackle such problems under every possible setting. This stemmed various approaches to uncertainty, resulting in models differing in complexity, applicability, and other qualitative criteria. To name a few, let us highlight fuzzy cooperative games~\cite{branzei2008models,mares2013fuzzy,fuzzy}, multi-choice games~\cite{branzei2008models}, cooperative interval games~\cite{gokphd,alparslan2011set,BH15}, fuzzy interval games~\cite{fuzzyinterval}, games under bubbly uncertainty~\cite{palanci2014cooperative}, ellipsoidal games~\cite{weber2010cooperative}, and games based on grey numbers~\cite{palanci2015cooperative}.

In the classical cooperative game theory, worth of each group of players (called
\emph{coalition}) is known. This represents the reward for a possible
cooperation of its members. For incomplete cooperative games, this is
no longer true, since in general values for only some of the coalitions are known,
while the others are not.

The model was first introduced by
Willson~\cite{willson1993value} in 1993 (instead of incomplete games, Willson called them \emph{partially defined games}). He gave the basic notion of
incomplete games and he generalised the definition of the Shapley value for such
games. More than two decades later, Inuiguchi and Masuya continued in this line of research~\cite{Masuya2016}. Their main focus was on superadditivity of possible extensions of the underlying minimal incomplete games (called there \emph{incomplete games with minimal information}).
%Subsequently, Masuya further extended the results in~\cite{masuya2016shapley,masuya2017axioms}, where he discussed more general classes of incomplete games and concentrated his efforts on generalisations of the Shapley value.
Recently, Bok, Černý, Hartman, and Hladík~\cite{BCHH20} extended the study by considering in detail properties of convexity and positivity in this model. Further, Masuya~\cite{masuya2021approximated} considered approximations of the Shapley value for incomplete games, and Yu~\cite{xiaohui2021extension} introduced a generalisation of incomplete games to games with coalition structures and studied the proportional Owen value (which is a generalisation of the Shapley value for these games).

\subsection{Our results and structure of the paper}

Let us highlight our results and provide a guide over the structure of this paper.

\begin{itemize}
	\item Section~\ref{sec:motivation} provides motivations for our research, while Section~\ref{sec:preliminaries} contains preliminaries on convex sets, classical cooperative games, and incomplete games.
	\item Section~\ref{sec:minimal} focuses on minimal incomplete games.
	\begin{itemize}
	  	\item Subsection~\ref{subsec:minimal-description} gives a characterisation of 1-convex extendability (Theorem~\ref{thm:excess-extending}). We describe the upper game, extreme games and the set of 1-convex extensions (Theorems~\ref{thm:1c-upper}, \ref{conj:extreme-games}, and \ref{thm:oneC-min-info-set}, respectively).
	  	\item Regarding values of incomplete games, generalisations of the $\tau$-va\-lue, the Shapley value and the nucleolus are then considered. We compare different variants and prove that these values actually coincide on the set of 1-convex extensions of minimal incomplete games (Subsection~\ref{subsec:minimal-values} and Theorems~\ref{thm:t-value-formula}, \ref{thm:conic-equals-average}, \ref{thm:ave-shapley-formula} therein).
			\item This enables us to introduce notion of the average value for which we obtain several characterisations and axiomatisations involving some natural properties (Subsection~\ref{subsec:minimal-av} and Theorems~\ref{thm:axiom-known-tau}, \ref{thm:axiom-known-shapley}, \ref{thm:ave-value-axiom1}, \ref{thm:ave-value-axiom2} therein). 
	  \end{itemize}
	\item Section~\ref{sec:upper} introduces games with defined upper vector and in the similar fashion to the minimal incomplete games, we first give a characterisation of 1-convex extendability (Theorem~\ref{thm:duv-extendability}). We use it to provide descriptions of extreme games and the set of 1-convex extensions (Subsection~\ref{subsec:upper-description} and Theorems~\ref{thm:duv-extreme} and \ref{thm:duv-description}). Most importantly, based on these results, we show that the analogously defined variants of the Shapley value do not coincide in general for incomplete games with defined upper vector with at least four players.
	\item Section~\ref{sec:conclusion} concludes with a brief summary and possible directions for future research.
\end{itemize}

\section{Motivation}\label{sec:motivation}

We would like to highlight a few areas which directly motivated our research, apart from the already mentioned and immediate motivations regarding uncertainty.

\paragraph{Connections to other models of uncertainty.}
In Definition~\ref{def:lower-upper-games}, we introduce two games bounding a set of extensions. The bounding games can be interpreted as a representation of a cooperative interval game in which coalitions have their values represented by closed real intervals.
Furthermore, the set of all 1-convex extensions forms a subset of all realisations (\emph{selections}~\cite{gokphd,BH15}) of the interval game and this interval game is in a sense the tightest among all possible ones. Since cooperative interval games are in fact a special case of interval fuzzy cooperative games, the studied model is subsequently also connected to these.

\paragraph{Set functions.} \emph{Set functions} are functions whose domain is a family of subsets of a given set. Incomplete games can be viewed as partial set functions. Extendability of partial set functions was extensively
studied~\cite{seshadhri2014submodularity,BK2018partextension,BK2019coverextension} (although without highlighting the connection to the theory of incomplete games).
We refer to the book of Grabisch~\cite{grabisch2016setfunctions}
which discusses exhaustively connections of various types of set
functions to entirely different parts of mathematics, with cooperative game theory being one of them.

\paragraph{Computational aspects.} In general, the domain of a characteristic function of a cooperative game with $n$ players is of size $2^n$. For incomplete games, the domain can be smaller which can be used as an advantage. In this paper, the size of the domain is even linear with respect to the number of players. Under incomplete information, it is impossible to compute solution concepts exactly. Our results often directly imply efficient  algorithms for finding an approximation. The other way around, one can also sample just a portion of coalition values from a given classical cooperative game and do effective computations just on this sample (with results being an approximation), thus in fact working with incomplete games.

\paragraph{1-convexity and core-catchers.} Core, one of the well-known measures of stability in cooperative games, is often hard to find both in terms of a compact description and computational complexity. This initiated a study of the so called \emph{core-catchers} being defined as a specific superset of the core with a succinct description. Non-emptiness of a core-catcher can be used as an easily verifiable necessary condition for non-emptiness of core. The class of 1-convex games is a subset of \emph{quasi-balanced games} and \emph{semi-balanced games} \cite{Tijs1981,tijs1982hypercube} which are both defined as classes of games with a respective non-empty core-catcher.

\paragraph{1-convexity and practical applications.} The importance of 1-convexity is not solely theoretical. In fact, a recent study of Dehez~\cite{dehez2021} surveys applications of 1-convex games, namely in bankruptcy resolution, joint subscription problem, sharing data and patents situations, and in provision of an indivisible public good (see references therein). Apart from that, cores of information games~\cite{driessen1992cooperative} and co-insurance games~\cite{Driessen2012} are single allocations under 1-convexity.

\section{Preliminaries}\label{sec:preliminaries}

We now introduce fundamentals of the theory of convex sets, classical cooperative games, and incomplete games. We present only the necessary background needed for our study of 1-convexity in the framework of incomplete cooperative games. 

We often replace singleton set $\left\{i\right\}$ with just $i$. We use $\subseteq$ for the relation of
``being a subset of'' and $\subsetneq$ for the relation ``being a proper subset
of''. To denote the sizes of sets e.g. $N,S,T$, we often use $n,s,t$, respectively.
We often write $x(S)=\sum_{i \in S} x_i$ for vectors $x \in \Rn$.
For an inequality $L(x)\leq R(x)$, where $L(x)$ and $R(x)$ are the left-hand side and the right-hand side in variable $x \in \Rn$, respectively, we distinguish two cases. For $x^* \in \Rn$, the inequality is \emph{strict} (at $x^*$) if $L(x^*) < R(x^*)$ and it is \emph{tight} or \emph{binding at $x^*$} if $L(x^*)=R(x^*)$. 

\subsection{Convex sets}
We state all the results in this subsection as facts and refer the reader to the book by Soltan~\cite{Soltan2015} with an exhaustive analysis of convex sets.

A \emph{polyhedron} $P$ in $\Rn$ is the intersection of finitely many halfspaces or equivalently a set $P \coloneqq \{x \in \R^n \mid Ax \leq b\}$ for a matrix $A \in \mathbb{R}^{m \times n}$ and a vector $b\in \Rn$. Polyhedron $P \subseteq \Rn$ is \emph{bounded from above} if there is $x \in \Rn$ satisfying for every $y \in P$, $y \leq x$. It is \emph{bounded from below} if there is $z \in \Rn$ satysfying for every $y \in P$, $y \geq z$. Another way to express polyhedrons is using its \emph{extreme points} and \emph{extreme rays}.

\begin{definition}\label{thm:extreme-points}
	Let $P \subseteq \R^n$ be a polyhedron. Then $e \in P$ is an \emph{extreme point} (of $P$) if for every $x \in \R^n$, we have $(e+x) \in P \wedge (e-x) \in P \implies x = 0$.
\end{definition}

For the definition of extreme rays, we employ the \emph{recession cone}.

\begin{definition}
	Consider a nonempty polyhedron $P = \{x \in \mathbb{R}^n \mid Ax \leq b\}$ and $y \in P$. The \emph{recession cone} of $P$ (at $y$) is the set
	\[
	R \coloneqq \{d \in \mathbb{R}^n \mid y + \lambda d \in P \text{ for all } \lambda \geq 0\}.
	\]
\end{definition}

By definition, the recession cone consists of all directions along which we can move indefinitely from $y$ and still remain in $P$. Notice that $y + \lambda d \in P$ for all $\lambda \geq 0$ if and only if $A(y + \lambda d) \leq b$ for all $\lambda \geq 0$ and this holds if and only if $Ad \leq 0$. Thus $R$ does not depend on a specific vector $y$ and can be expressed as $R = \{d \in \Rn \mid Ad \leq 0\}$. 

\begin{definition}\label{thm:extreme-rays}
	A nonzero element $x$ of a recession cone $R =\{d \in \Rn \mid Ad \leq 0\}$ is an extreme ray (of $R$) if there are $n-1$ linearly independent constraints binding at $x$, i.e.\ 
	\[\text{there are } k_1,\dots,k_{n-1} \in \left\{1,\dots,n\right\} \text{ such that } A_{k_{i}*}x=0,\]
	where $A_{i*}$ represents $i$-th row of matrix $A$.
\end{definition}

\begin{definition}
	An \emph{extreme ray} of $P$ is any extreme ray of its recession cone $R$.
\end{definition}

The following theorem gives a full description of pointed polyhedron (i.e.\ unbounded convex set with at least one extreme point) based only on its extreme points and extreme rays.

\begin{theorem}\label{thm:char-of-polyhedrons}\cite{Soltan2015}
	Let $P = \{x \in \mathbb{R}^n \mid Ax \leq b\}$ be a nonempty polyhedron with at least one extreme point. Let $x_1,\dots,x_r$ be its extreme points and $y_1,\dots,y_\ell$ be its extreme rays. Then
	\[
		P = \left\{\sum\limits_{i=1}^r\alpha_ix_i + \sum\limits_{j=1}^\ell\beta_iy_i \mid \forall i,j: \alpha_i \geq 0 , \beta_j \geq 0,  \sum\limits_{i=1}^r\alpha_i=1\right\}.
	\]
\end{theorem}

\subsection{Classical cooperative games}

Here we define cooperative games and in particular 1-convex games. We proceed with an introduction of value functions, namely the $\tau$-value, the nucleolus and the Shapley value and we review their properties for 1-convex games. We invite the interested reader to consult the following sources on cooperative games for more:~\cite{branzei2008models,driessen1988cooperative,gilles2010cooperative,peleg2007introduction}. For heavier focus on applications, see e.g.~\cite{bilbao2012cooperative,combopt,insurance}.

\subsubsection{Main definitions and notation}

% We note that the presented definition assumes transferable utility (shortly TU). Therefore, by a cooperative game or a game we mean in fact a cooperative TU game.

\begin{definition}
	A \emph{cooperative game} is an ordered pair $(N,v)$ where the set $N = \{1,2,\ldots ,n\}$ and $v\colon 2^N \to \mathbb{R}$ is the characteristic function of the cooperative game. Further, $v(\emptyset) = 0$.
\end{definition}

We denote the set of $n$-person cooperative games by $\Gamma^n$. Subsets of $N$ are called \emph{coalitions} and $N$ itself is called the
\emph{grand coalition}. We often write $v$ instead of $(N,v)$ whenever there
is no confusion over what the player set is. We often associate the characteristic functions $v\colon 2^N\to\mathbb{R}$ with vectors $v\in\mathbb{R}^{2^n-1}$. This is convenient for viewing sets of cooperative games as convex sets of points.

The \emph{zero-normalised game} of $(N,v)$ is the game $(N,v_0)$ defined as $v_0(S) \coloneqq v(S) - \sum_{i \in S} v(i)$ for every $S \subseteq N$. Observe that $v_0(i) = 0$ for every $i \in N$. Analogously, we say that a game is \emph{zero-normalised} if all singletons have value equal to zero. 

The definition of 1-convex games relies on the notion of \emph{upper vector} (or \emph{utopia vector}) $b^v \in \mathbb{R}^n$ for a given game $(N,v)$. It captures each player's marginal contribution to the grand coalition, i.e. $b^v_i \coloneqq  v(N) - v(N\setminus i)$. When there is no ambiguity, we tend to use just $b$ instead of $b^v$. The value $b^v_i$ is considered to be the maximal rightful value that player $i$ can claim when $v(N)$ is distributed among players. If he claims more, it is more advantageous for the rest of the players to form a coalition without player $i$.

\begin{definition}\label{def:1conv-cond}
	A cooperative game $(N,v)$ is called \emph{1-convex game}, if for all coalitions $\emptyset \neq S \subseteq N$, the inequality
	\begin{equation}\label{def:1conv-cond1}
		v(S) \leq v(N) - b(N\setminus S)
	\end{equation}
	holds and also
	\begin{equation}\label{def:1conv-cond2}
		b(N) \geq v(N).
	\end{equation}
	The set of 1-convex $n$-person games is denoted by $\oneC$.
\end{definition}

By~\eqref{def:1conv-cond1}, $(N,v)$ is 1-convex if even after every player outside the coalition $S$ gets paid his utopia value, there is still more left of the value of the grand coalition for players from $S$ than if they decided to cooperate on their own. This condition challenges the players to remain in the grand coalition and try to find a compromise in the payoff distribution. Also, in~\eqref{def:1conv-cond2}, the utopia vector sums to a value at least as large as the value of the grand coalition $N$. This was motivated by the idea that the study of possible distributions is not interesting if every player can obtain his maximal rightful (utopia) value.

An equivalent formulation of 1-convexity is in terms of the \textit{gap function}, defined as $g^v(S) \coloneqq b(S) - v(S)$. It captures the gap between the utopia distribution for coalition $S$ and a possible distribution of the value of $S$.

\begin{theorem}~\cite{Driessen1985}
	A game $(N,v)$ is 1-convex if and only if $0 \leq g(N) \leq g(S)$ for all coalitions $S\subseteq N$.
\end{theorem}

Intuitively, the grand coalition is in some sense closest to the utopia distribution among possible coalitions.
We can also rewrite conditions~\eqref{def:1conv-cond1} and~\eqref{def:1conv-cond2} in terms of the characteristic function as follows. We provide such equivalent formulations as they will be useful in later proofs.
For $\emptyset \neq S \subseteq N$,
\begin{equation}\label{def:1conv-cond1-alt}
	v(S) + (n-s-1)v(N) \leq \sum_{i \in N \setminus S} v(N \setminus i),
\end{equation}
and
\begin{equation}\label{def:1conv-cond2-alt}
	(n-1)v(N) \geq \sum_{i \in N} v(N\setminus i).
\end{equation}

\subsubsection{Values of complete games}
The main task of cooperative game theory is to distribute the value of the grand coalition $v(N)$ between the players. To be able to work with individual payoffs more easily, \emph{payoff vectors} are introduced. Those are vectors $x \in \Rn$ where $x_i$ represents the individual payoff of player $i$.

To choose between payoff vectors, different \emph{solution concepts} assigning some subset of all possible payoff vectors are defined. In this paper, we focus solely on \emph{single-point solution concepts}, sometimes referred to as \emph{values}.

\begin{definition}
	Let $C\subseteq \Gamma^n$ be a class of $n$-person cooperative games. Then a function $f\colon C\to\Rn$ is a \emph{value} (on class $C$).
\end{definition} 

For a cooperative game $(N,v)$, $f(v)$ is called the value of $(N,v)$.

We consider a generalization of two (actually three) values: the \emph{$\tau$-value}, the \emph{nucleolus}, and the \emph{Shapley value}. We now introduce these values, stating their properties and different characterisations, which is used for our generalisations to incomplete games.

\paragraph{The $\tau$-value.} The $\tau$-value is a well-known solution concept for 1-convex games (actually defined for a more general class of quasi-balanced games) introduced by Tijs in 1981~\cite{Tijs1981}. We follow his definition where he defines the $\tau$-value as a compromise between the utopia vector $b^v$ and the \textit{minimal right vector} $a^v$ that is defined as $a^v\coloneqq b^v - \lambda^v$ where $\lambda^v$ is the so called \textit{concession vector} defined as $$\lambda^v_i \coloneqq  \min\limits_{S \subseteq N, i \in S}g^v(S).$$

The class of quasi-balanced games $Q^n$ is defined as \[Q^n\coloneqq \{(N,v) \mid \forall i \in N: a_i^v \leq b_i^v\text{ and } a^v(N) \leq v(N) \leq b^v(N)\}.\]
It readily holds that $\oneC \subseteq Q^n$.
\begin{definition}
	Let $(N,v) \in Q^n$. Then the \emph{$\tau$-value} $\tau(v)$ of game $(N,v)$ is defined as the unique convex combination of $a^v$ and $b^v$ such that $\sum\limits_{i \in N}\tau(v)_i = v(N)$. 
\end{definition}
For class $\oneC$, the $\tau$-value can be expressed by a simple formula depending on the utopia vector and the gap function. The formula can be interpreted as follows. Every player receives his utopia value minus an equal share of the loss represented by the gap $g(N)$.

\begin{theorem}\label{thm:1conv-formula}~\cite{Driessen1985-thesis}
	Let $(N,v)$ be a 1-convex game. Then the $\tau$-value can be expressed as
	\[
	\tau_i(v) = b^v_i - \frac{g^v(N)}{n},\forall i \in N.
	\]
\end{theorem}

There are two known axiomatic characterisations of the $\tau$-value.

\begin{theorem}\label{thm:tau-axiom1}\cite{Tijs1987}
	The $\tau$-value is a unique function $f\colon Q^n\to\mathbb{R}^n$ satisfying the following axioms for every $v \in Q^n$:
	\begin{enumerate}
		\item (\emph{efficiency}) $\sum_{i\in N}f_i(v) = v(N)$,
		\item (\emph{minimal right property}) $f(v) = a^v + f(v-a^v)$,\\
		where $(v-a^v)(S) = v(S)-\sum_{i \in S}a_i^v$,
		\item (\emph{restricted proportionality property}) $f(v_0) = \alpha b^{v_0}$\\ where $\alpha \in \mathbb{R}$ and $(N,v_0)$ is the zero-normalised game of $(N,v)$.
	\end{enumerate}
\end{theorem}
The second, axiomatic characterisation can be found in paper of Tijs~\cite{Tijs1995}. It consists of five axioms, namely \emph{efficiency}, \emph{translation equivalence}, \emph{bounded aspirations}, \emph{convexity}, and \emph{restricted linearity}.

On top of that, there are further results concerning axioms of the $\tau$-value, thus providing an even better comparison with other values.
% In the next theorem, we state several of them.

\begin{theorem}\label{thm:t-value-props}\cite{Tijs1981}
	For a 1-convex game $(N,v)$, the $\tau$-value $\tau(v)$ satisfies:
	\begin{enumerate}
		\item (\textit{individual rationality}) $\forall i \in N: \tau_i(v) \geq v(i)$,
		\item (\textit{efficiency}) $\sum_{i \in N}\tau_i(v) = v(N)$,
		\item (\textit{symmetry}) for each permutation $\pi\colon N \to N$, we have $\tau(\pi_*v) = \pi_*(\tau(v))$,
		\item (\textit{dummy player}) $\forall i \in N, \forall S \subseteq N: v(S \cup i) = v(S) \implies \tau_i(v)=0$,
		\item (\textit{S-equivalence property}) $\forall k\in [0,\infty], \forall c\in\mathbb{R}:\tau(k\cdot v + c) = k\cdot\tau(v) + c$. 
	\end{enumerate}
\end{theorem}

We note the $\tau$-value does not satisfy \emph{additivity} which is crucial in our generalisation of this concept. Surprisingly, we show that our generalisation of the $\tau$-value satisfies a certain form of additivity on the class of minimal incomplete games.

\paragraph{The nucleolus.} The essential component of the definition of nucleolus is the \textit{excess} $e(S,x)$ which is a function dependent on a coalition $S$ and an \textit{imputation} $x$ --- a payoff vector which is both individually rational ($x_i \geq v(i)$ for all $i \in N$)  and efficient ($x(N) = v(N)$). It computes the remaining potential of $S$ when the payoff is distributed according to $x$, i.e.\, $e(S,x) \coloneqq  v(S) - x(S)$. Further, $\theta(x) \in \mathbb{R}^{2^n}$ is a vector of excesses with respect to $x$ which is arranged in non-increasing order.

\begin{definition}
	The \emph{nucleolus}, $\eta\colon \Gamma^{n}\to \mathbb{R}^n$, is the solution concept which assigns to a given game the minimal imputation $x$ with respect to the lexicographical ordering $\theta(x)$, defined as:
	\[
	\theta(x) < \theta(y) \text{ if } \exists k: \forall i < k: \theta_i(x)=\theta_i(y) \text{ and } \theta_k(x) < \theta_k(y).
	\]  
\end{definition}

A basic result in cooperative game theory states that the nucleolus is a one-point solution concept (value)~\cite{schmeidler1969nucleolus}.
In general, the nucleolus can be computed by means of linear programming~\cite{Kopelowitz1967}. For 1-convex games, however, the notion of the nucleolus and the $\tau$-value coincide.

\begin{theorem}\cite{Driessen1985-thesis}
	Let $(N,v)$ be 1-convex game. Then $\eta(v)=\tau(v)$.
\end{theorem}

In this text, we consider a generalisation of the $\tau$-value for $\oneC$-extendable incomplete games. However, thanks to the theorem it can be also considered as a generalisation of the nucleolus.

\paragraph{The Shapley value.} The last of the studied values is the \emph{Shapley value}.

\begin{definition}\label{def:shapley}
	The \emph{Shapley value} is the function $\phi \colon \Gamma^n \to \mathbb{R}^n$ such that
	\[
	\phi_i(v) \coloneqq \sum\limits_{S \subseteq N, i \in S}\frac{(\lvert S \rvert - 1)!(n - \lvert S \rvert)!}{n!}(v(S)-v(S\setminus i)),\forall i \in N.
	\]
\end{definition}

There are alternative formulas for the Shapley value, including the following one (we use it later in Theorem~\ref{thm:ave-shapley-formula}).

\begin{theorem}\label{thm:alternate-shapley-formula}~\cite{Peters2008}
	The Shapley value for $(N,v)$ can be expressed as follows:
	$$\phi_i(v) = \frac{1}{n}\sum\limits_{S\subseteq N \setminus i}{n-1\choose s}^{-1}(v(S\cup i) - v(S)),\forall i \in N.$$
\end{theorem}

The Shapley value can be also characterised by means of axioms. The following is the characterisation proposed and proved by Shapley.

\begin{theorem}~\cite{Shapley1953}
	The Shapley value is a unique function $f\colon \Gamma^n\to\mathbb{R}^n$ satisfying the following for every $v,w \in \Gamma^n$:
	\begin{enumerate}
		\item (\textit{efficiency}) $\sum_{i\in N}f_i(v) = v(N)$,
		\item (\textit{symmetry}) $\forall i,j\in N, \forall S \subseteq N \setminus \{i,j\}: v(S+i)=v(S+j) \implies f_i(v)=f_j(v)$,
		\item (\textit{null player}) $\forall i \in N,\forall S \subseteq N:v(S) = v(S+i) \implies f_i(v)=0$,
		\item (\textit{additivity}) $f(v+w)=f(v) + f(w)$.
	\end{enumerate}
\end{theorem}

Since the original introduction of the Shapley value, many alternative axiomatic characterisations of the Shapley value were given. Let us pinpoint the following few: \cite{Brink1995,Brink2002,Roth1977,Young1989}. As it would be an exhaustive task to investigate all of them at once, we considered only some of them (namely the second and the fourth mentioned). The Shapley value also satisfies all of the axioms from Theorem~\ref{thm:t-value-props} except for individual rationality.

\subsection{Incomplete cooperative games}

\begin{definition}\emph{(Incomplete game)}
	An incomplete game is a tuple $(N,\K,v)$ where the set $N = \{1,\dots,n\}$, $\K \subseteq 2^N$ is the set of coalitions with known values and $v\colon \mathcal K \to \R$ is the characteristic function of the incomplete game. Further, $\emptyset \in \K$ and $v(\emptyset)=0$.
\end{definition}

We denote the set of $n$-person incomplete games with $\K$ by $\Gamma^n(\K)$. An incomplete game can be viewed from several perspectives. In one of the views, there is an underlying complete game $(N,v)$ from a class of $n$-person games $C \subseteq \Gamma^n$. The presence of $(N,v)$ in $C$ implies further properties of the characteristic function, e.g. superadditivity. Unfortunately, only partial information (captured by $(N,\K,v)$) is known and there is no way to acquire more knowledge. The goal is then to reconstruct $(N,v)$ as accurately as possible. This leads to the definition of $C$-extensions.

\begin{definition}
	Let $C$ be a class of $n$-person games. A cooperative game $(N,w) \in C$ is a \emph{$C$-extension} of an incomplete game $(N,\K,v)$ if $w(S)=v(S)$ for every $S \in \mathcal K$. 
\end{definition}

The set of all $C$-extensions of an incomplete game $(N,\K,v)$ is denoted by $C(v)$. We write \emph{$C(v)$-extension} whenever we want to emphasize the game $(N,\K,v)$. Also, if there is a $C(v)$-extension, we say $(N,\K, v)$ is \emph{$C$-extendable}. Finally, the set of all $C$-extendable incomplete games with fixed $\K$ is denoted by $C(\K)$. In this text, we are mainly interested in $\oneC$-extensions. 

The sets of $C$-extensions studied in this text are always convex. One of the main goals of the model of incomplete cooperative games is to describe sets of $C$-extensions by using their extreme points and extreme rays whenever the description is possible. We refer to the extreme points as to \emph{extreme games}.

If the structure of $C(v)$ is too difficult to describe and it is bounded from above, we introduce the the upper game.

\begin{definition}\label{def:lower-upper-games}
	\emph{(The upper game of a set of $C$-extensions)} Let
	$(N,\mathcal K,v)$ be a $C$-extendable incomplete game. If $C(v)$ is bounded from above, then the \emph{upper game} $(N,\overline{v})$ of $C(v)$ is a complete game such that for every $(N,w) \in C(v)$ and for every $S \subseteq N$, we have $$w(S) \leq \overline{v}(S).$$
	Furthermore, for every $S \subseteq N$, there is $(N,w') \in C(v)$ such that 
	$$\overline{v}(S) = w'(S).$$
\end{definition}

One can analogously define also \emph{the lower game of a set of $C$-extensions}.

These games delimit the area of $\mathbb{R}^{2^n}$ that contains the set of $C$-extensions. Even if we know the description of $C(v)$, the lower and the upper game are still useful as they encapsulate the range of possible values $\left[\underline{v}(S),\overline{v}(S)\right]$ of coalition $S$ across all possible $C$-extensions.

In many situations in the cooperative game theory, full information on a cooperative game is not necessary for a satisfiable answer. For example, the $\tau$-value of a 1-convex cooperative game $(N,v)$ depends only on values $v(N)$ and $v(N\setminus i)$ for $i \in N$. What if there are other satisfiable ways to distribute the payoff between players that can be computed only from partial information encoded by an incomplete game? Based on this question, we can generalise values to incomplete games.

\begin{definition}
	Let $C(\K)$ be a class of $C$-extendable $n$-person incomplete games. Then function $f\colon C(\K)\to{\mathbb{R}^{n}}$ is a \emph{value} (on class $C(\K)$).
\end{definition} 

For an incomplete game $(N,\K,v)$, $f(v)$ is called the value of $(N,\K,v)$. The model of incomplete cooperative games is still in its beginnings. One of the most significant downsides of classical cooperative games is the complexity of information required. For an $n$-person game, we have to consider $2^n$ different coalitions with corresponding values of the characteristic function and to be able to apply the model, we often need all this information (the $\tau$-value of 1-convex games is rather an exception).

\section{Minimal incomplete games}\label{sec:minimal}

In this section, we focus on the subclass of minimal incomplete games and their $\oneC$-extensions. A \emph{minimal incomplete game} $(N,\K,v)$ is an incomplete game that contains (apart from $\emptyset$) only the grand coalition and singletons, i.e. $\K = \{\{i\}\mid i \in N\} \cup \{\emptyset,N\}$. We denote such $\K$ as $\Kmin$ further in the text (the number of players should be clear from the context). We also define the \emph{total excess} as $\Delta \coloneqq v(N) - \sum_{i \in N}v(i)$ which is widely used further in this text.

Analogously to classical cooperative games, for any $v \in \oneC(\Kmin)$, the \emph{zero-normalised game} $v_0$ can be described as
	\[
	v_0(S) = \begin{cases}
		v(N)-\sum_{i \in N}v(i), & \text{if } S = N,\\
		0, & \text{if } S \neq N.\\
	\end{cases}
	\] 

Usually, minimal incomplete games promise a simpler analysis of the set of extensions compared to the general case due to its simple structure of coalitions with known values. However, its purpose is not purely theoretical. Consider situations where knowledge of values of individual players and of the value of all the players together is easily accessible, while for other coalitions, an extra effort has to be made. This effort can represent for examples making a survey or a more thorough analysis of the situation, difficulties regarding negotiations, revealing intentions or private data, or some other administrative and logistical steps that have to be taken before one can get the value such coalitions.

In the first subsection, we derive a description of the set of $\oneC$-extensions. In the second subsection, we define different values and show they coincide on minimal incomplete games. Calling it \emph{the average value}, we investigate its axiomatisations in the third subsection.

\subsection{Description of the set of $\oneC$-extensions}\label{subsec:minimal-description}

The first step towards understanding the set of $\oneCv$-extensions is to characterise its emptiness.

\begin{theorem}\label{thm:excess-extending}
	A minimal incomplete game $(N,\Kmin,v)$ is $\oneC$-extendable if and only if $\Delta \geq 0$.
\end{theorem}
\begin{proof}
	Let $(N,w) \in \oneCv$. Since it is 1-convex, it must hold for each $i \in N$,
	\[
	w(i) \leq w(N) - b(N\setminus i).
	\]
	We sum the inequalities over all $n$ players to get
	\[
	\sum_{i \in N}w(i) \leq nw(N) - \sum_{i \in N}b^w(N\setminus i).
	\]
	We now expand expressions of type $b^w(N\setminus i)$ and rearrange the inequality into
	\begin{equation}\label{eq:oneC-extendability}
		\sum_{i \in N}w(i) + n(n-2)w(N) \leq (n-1) \sum_{i \in N}w(N \setminus j).
	\end{equation}
	Since $b^w(N) \geq w(N)$ is equivalent to $\sum_{i \in N}w(N \setminus i) \leq (n-1)w(N)$, we bound the right-hand side of \eqref{eq:oneC-extendability} by $(n-1)^2w(N)$ and by rearranging, we conclude that $\Delta \geq 0$.
	
	For the opposite direction, let us consider $\oneC$-extensions $(N,v^i)$ for $i \in N$ defined as
	\begin{equation}\label{eq:oneC-vertex-games}
		v^i(S) \coloneqq 
		\begin{cases}
			v(S), & \text{if } S \in \Kmin,\\
			v(N) - \sum_{j \not\in S}v(j), & \text{if } S \not\in \Kmin \wedge i \in S,\\
			v(N) - \sum_{j \not\in S}v(j) - \Delta, & \text{if } S \not\in \Kmin \wedge i \not\in S.
		\end{cases}
	\end{equation}
	Notice that such games coincide on values of $S\in \Kmin$. We claim that for any $i \in N$, the game $v^i \in \oneCv$.
	To verify Condition~\eqref{def:1conv-cond1} for every $S \subseteq N$, we distinguish two cases.
	
	It is easy to derive that for $i \in S$, $v^i(S)$ and $v^i(N) - b^{v^i}(N \setminus S)$ are equal to $v(N)-\sum_{j \notin S}v(j)$ and for $i \not\in S$, $v^i(S)$ and $v^i(N) - b^{v^i}(N \setminus S)$ are both equal to $v(N)-\sum_{j \notin S}v(j) - \Delta$. Thus, Condition~\eqref{def:1conv-cond1} holds for every $S \subseteq N$ with equal sign.

	Condition~\eqref{def:1conv-cond2} holds since
	\begin{equation*}
		\begin{aligned}
			b^{v^i}(N) &= n  v^i(N) - \sum_{j \in N} v^i(N \setminus j) \\
			&= n  v(N) - (n-1)\left(\sum_{i \in N}v(i) + \Delta\right) = v(N).
		\end{aligned}
	\end{equation*}
	We checked both conditions, which concludes the proof. \qed
\end{proof}

We note that if $\Delta=0$, the set of $\oneC$-extensions is rather simple and consists only of $(N,\overline{v})$ (the upper game defined in Theorem~\ref{thm:1c-upper}). Therefore, we are naturally more interested in situations when $\Delta > 0$.

The set of $\oneC$-extensions is not bounded from below. For a $\oneC$-extendable incomplete game $(N,\Kmin,v)$ and its $\oneCv$-extension $(N,w)$, we can construct yet another $\oneCv$-extension $(N,w_S)$ dependent on a coalition $S\subseteq N$ such that $1 < \lvert S \rvert < n-1$. We set the characteristic functions of the two games to differ only in values of $S$, so that $w_S(S) < w(S)$. The 1-convexity of $(N,w_S)$ is easy to check from 1-convexity of $(N,w)$ and it can be immediately seen that any arbitrarily small number $\varepsilon$ satisfying $\varepsilon < v(S)$ could be chosen for the value of coalition $S$ in $(N,w_S)$. While not bounded from below, the set of $\oneC$-extensions is bounded from above.

\begin{theorem}\label{thm:1c-upper}
	Let $(N,\Kmin,v)$ be a minimal $\oneC$-extendable game. Then the upper game $(N,\overline{v})$ has the following form:
	\[
	\overline{v}(S) \coloneqq \begin{cases}
		v(S), & \text{if } S \in \Kmin,\\
		v(N) - \sum_{i \not\in S} v(i), & \text{if } S \not\in \Kmin.
	\end{cases}
	\]
\end{theorem}
\begin{proof}
	
	To show that this is an upper bound for the value of each coalition $T \subseteq N$, we formulate the following optimization problem:
	\begin{equation}\label{eq;1c-upper-first}
		\begin{aligned}
			& \underset{(N,w) \in \oneCv}{\max}
			& & w(T) \\
			& \text{s.t.} & &  w(N) \leq b^w(N), \\
			& & &  w(S) \leq w(N) - b^w(N \setminus S) \text{ for } S \subseteq N, S \neq \emptyset. \\
		\end{aligned}
	\end{equation}
	Clearly, the optimal value of the optimization problem (if it exists) is the value $\overline{v}(T)$. Also notice that from the condition for $T$, i.e.\ $w(T) \leq w(N) - b^w(N\setminus T)$, that the upper bound of $w(T)$ is dependent only on value $w(N)$ (which is a constant since $N \in \Kmin$) and $n$ values $w(N \setminus i)$ for $i \in N$ (which are variables). The sum of these variables is bounded from above by $(n-1)v(N)$ (since $w(N) \leq b^w(N) \iff \sum_{i \in N}v(N \setminus i) \leq (n-1)v(N)$). From below, we have to consider only conditions $w(i) \leq w(N) - b^w(N \setminus i)$, because for $S \not\in \Kmin$, we can always choose a $\oneC$-extension such that the value $w(S)$ is small enough to satisfy $w(S) \leq w(N) - b^w(N \setminus S)$.
	
	Therefore, we can simplify the optimization problem by:
	\begin{enumerate}
		\item removing conditions for $S \not\in \Kmin$,
		\item removing variables $w(S)$ for $S \not\in \Kmin$, and
		\item substituting objective function $w(T)$ for $w(N) - b^w(N \setminus T)$.
	\end{enumerate}
	By these simplifications, we get the following optimization problem:
	\begin{equation}\label{eq;1c-upper-second}
		\begin{aligned}
			& \underset{}{\text{max}}
			& & w(N) - b^w(N \setminus T) \\
			& \text{s. t.} & &  w(i) \leq w(N) - b^w(N \setminus i) \\
			& & &  i = 1,\dots,n. \\
		\end{aligned}
	\end{equation}
	The set of feasible solutions is now $w \in \mathbb{R}^n$ where $w_i=w(N\setminus i)$ and $w(N)$ together with $w(i)$ (for $i \in N$) are constants. A feasible solution $w \in \mathbb{R}^n$ of problem~\eqref{eq;1c-upper-second} is equivalent to a feasible solution of problem~\eqref{eq;1c-upper-first} by setting $w(S) \coloneqq -(n-s-1)w(N) + \sum_{k \in N \setminus S}w_k = w(N)- b^w(N \setminus S)$. Notice that the optimal values for both problems with corresponding feasible solutions are equal.
	
	We restate the problem in terms of the characteristic function $w$ and we substitute $w(N \setminus i)$ for $w_i$, arriving at
	\begin{equation}\label{eq:1c-upper-third}
		\begin{aligned}
			& \underset{w \in \mathbb{R}^n}{\text{max}}
			& & \sum_{i \in N \setminus S}w_i - (n-s-1)w(N) \\
			& \text{s.t.} & &   \sum_{j \in N} w_j \leq (n-1)w(N) \\
			& & &  w(k) \leq \sum_{j \neq i}w_j - (n-2)w(N)\\
			& & &  k = 1,\dots,n.
		\end{aligned}
	\end{equation}
	
	Problem~(\ref{eq:1c-upper-third}) is an instance of linear programming. Therefore, we can construct its dual program:
	\begin{equation}\label{eq:1c-upper-fourth}
		\begin{aligned}
			& \underset{y \in \mathbb{R}^{n+1}}{\text{min}}
			& & \sum_{i \in N}\left[(-(n-2)w(N) - w(i))y_i\right] + (n-1)w(N)y_{n+1} - (n-s-1)w(N) \\
			& \text{s.t.} & & -\sum_{j \neq i}y_j + y_{n+1} = 1 \text{ for } i \not\in T \\
			& & &  -\sum_{j \neq i}y_j + y_{n+1} = 0 \text{ for } i \in T\\
			& & &   y_i \geq 0 \text{ } i = 1,\dots,n+1.
		\end{aligned}
	\end{equation}
	Let us define the vector $y^* \in \mathbb{R}^{n+1}$ as 
	\[
	y^*_j = \begin{cases}
		0, & \text{ if } j \in T,\\
		1, & \text{ if } j \not\in T,\\
		n-t, & \text{ if } j =n+1.
	\end{cases}
	\]
	We deduce that
	\begin{itemize}
		\item $y^*_j \geq 0$ for all $j \in N$,
		\item $-\sum_{j \neq i}y^*_j + y^*_{n+1} = -(n-t-1)1 + (n-t) = 1$ for $i \not\in T$,
		\item $-\sum_{j \neq i}y^*_j + y^*_{n+1} = -(n-t)1 + (n-t) = 0$ for $i \in T$.
	\end{itemize}
	Hence $y^*$ is a feasible solution of (\ref{eq:1c-upper-fourth}). Furthermore, the value of the objective function for $y^*$ equals $w(N)-\sum_{i \not\in T}w(i) = v(N)-\sum_{i \not\in T}v(i)$. This means (from the duality of linear programming) that the primal program is feasible and the value of its objective function is bounded from above by this value.
	
	To see that this upper bound is attained, take the game $(N,v^i)$ (defined in~\eqref{eq:oneC-vertex-games}) such that $i \not\in T$.
\qed\end{proof}

It is important (and by our opinion interesting) that the upper game of the set of superadditive extensions of non-negative minimal incomplete games~\cite{Masuya2016} coincides with the upper game of $\oneC$-extensions from Theorem~\ref{thm:1c-upper}.

The upper game $(N,\overline{v})$ is not 1-convex in general. For example, a 3-person minimal incomplete game $(N,\Kmin,v)$ with $v(\{1,2,3\})=1$ and $v(1)=v(2)=v(3)=0$ is $\oneC$-extendable since $\Delta = 1$. However, $1= \overline{v}(N) \nleq b^{\overline{v}}(N)=0$. The following theorem characterises when the upper game is 1-convex.

\begin{theorem}
	Let $(N,\Kmin,v)$ be a minimal incomplete game. Then it holds that the upper game $(N, \overline{v})$ (defined in Theorem~\ref{thm:1c-upper}) is 1-convex if and only if
	\[
	\Delta = 0 \text{ and } v(N) \leq \min\limits_{\emptyset \neq S \subsetneq N}\Bigg\{\frac{2}{n-s}\sum\limits_{i \in N \setminus S}v(i)\Bigg\}.
	\]
\end{theorem}
\begin{proof}
	From the condition $\overline{v}(N) \leq b^{\overline{v}}(N)$, one can derive $v(N)\leq\sum_{i \in N}v(i)$, and hence $\Delta = 0$. For $\emptyset \neq S \subsetneq N$ and conditions $\overline{v}(S)\leq \overline{v}(N) - b^{\overline{v}}(N\setminus S)$, can be expressed as
	\[v(N) - \sum_{i \in N \setminus S}v(i) \leq v(N) - \sum_{i \in N\setminus S}\left(v(N)-v(i)\right),\]
	or equivalently as
	\begin{equation}\label{eq:upper-game-is-1conv}
	v(N) \leq \frac{2}{n-s}\sum_{i \in N \setminus S}v(i).
	\end{equation}
	As for all $\emptyset \neq S \subsetneq N$, the left-hand side of~\eqref{eq:upper-game-is-1conv} is always $v(N)$, it suffices to enforce the condition with minimal right-hand side over all $S$.
\qed\end{proof}

So far, we showed that the set $\oneCv$ is a convex polyhedron since it can be described by a set of inequalities. It is bounded from above by $(N,\overline{v})$ and unbounded from below. Such polyhedrons (provided that they have at least one vertex) can be characterised by a set of extreme points and a cone of extreme rays (see Theorem~\ref{thm:char-of-polyhedrons}).

We begin the derivation of the full description of the set of $\oneC$-extensions by proving that games $(N,v^i)$ (defined as \eqref{eq:oneC-vertex-games}) are actually extreme points of the set. To prove this, we use a characterisation of extreme points from Theorem~\ref{thm:extreme-points}.

\begin{theorem}\label{thm:onec_extreme_games}
	For $\oneC$-extendable minimal incomplete game $(N,\Kmin,v)$, games $(N,v^i)$, defined in~\eqref{eq:oneC-vertex-games}, are extreme games of $\oneCv$.
\end{theorem}
\begin{proof}
	Let $x \in \mathbb R^{2^n}$ be a vector such that $(N,v^i + x), (N,v^i-x) \in \oneCv$. We aim to show that in such case, inevitably $x(S) = 0$ for all $S \subseteq N$, thus by Theorem~\ref{thm:extreme-points} $(N,v^i)$ is an extreme game.
	
	Denote $f^+ = v^i+x$ and $f^- = v^i-x$. For $S \in \Kmin$, clearly $x(S) = 0$. It remains to show $x(S)=0$ for $S \notin\Kmin$.
	
	For $S \not\in \Kmin$ and $i \not\in S$, for the sake of contradiction, suppose w.l.o.g. that $x(S) > 0$. Then $f^+(S) = v^i(S) + x(S) = \overline{v}(S) + x(S) > \overline{v}(S)$, therefore $(N,f^+)\not\in\oneCv$, a contradiction.
	
	For $S \not\in \Kmin$ and $i \in S$, because $(N,f^+)$ is 1-convex, it follows from~\eqref{def:1conv-cond1-alt} that
	\[
	f^+(S) + (n-s-1)f^+(N) \leq \sum_{j \not\in S}f^+(N\setminus j).
	\]
	We can rewrite the inequality as
	\[
	v^i(S) + x(S) + (n-s-1)\left(v(N) + x(N)\right) \leq \sum_{j \not\in S}\left(v^i(N \setminus j) + x(N \setminus j)\right).
	\]
	By expressing $\sum_{j \notin S}v^i(N \setminus j) = (n-s)v(N)-\sum_{j \in S}v(j)-\Delta$ and $v^i(S) = v(N) - \sum_{j \not\in S}v(j) - \Delta$, we can simplify the inequality to
	\[
	x(S) + (n-s-1)x(N) \leq x(N) + \sum_{j \notin S} x(N \setminus j)
	\]
	or, since $x(N)=0$,
	\[
	x(S) \leq \sum_{j \not\in S}x(N \setminus j).
	\]
	Similarly for $(N,f^-)$, we arrive at $x(S) \geq \sum_{j \notin S}x(N \setminus j)$, which together with the former inequality implies $x(S) = \sum_{j \not\in S}x(N \setminus j)$. Since $i \not\in S$, we have $i \in N \setminus j$ for every $j \notin S$. Thus $x(N \setminus j)=0$ as we prove above. This means $x(S) = \sum_{j \notin S}x(N \setminus j)=0$.
	
	We proved that $x$ is necessarily a vector of zeroes and thus we conclude the proof by taking Theorem~\ref{thm:extreme-points} into account. 
\qed\end{proof}

Not only are the games $(N,v^i)$ for $i \in N$ extreme games of $\oneCv$, they are also the only extreme games.

\begin{theorem}\label{conj:extreme-games}
	For a $\oneC$-extendable minimal incomplete game $(N,\Kmin,v)$, games $(N,v^i)$, defined in~\eqref{eq:oneC-vertex-games}, are the only extreme games of $\oneCv$.
\end{theorem}
\begin{proof}
	We prove this theorem by showing that any extreme game $(N,e)$ has the form of one of the $(N,v^i)$ games. Since there are $n$ different games, we have to enforce that the game coincides with $(N,v^i)$ for a specific $i$.
	
	To do so, realise there is $i$ such that $e(N \setminus i) < \overline{v}(N \setminus i)$. If there was no such $i$, then $\forall k: e(N\setminus k) \geq \overline{v}(N \setminus k)$. The sum of these conditions leads to
	\[
	\sum_{k \in N}e(N\setminus k) > \sum_{k \in N}\overline{v}(N \setminus k) = n \overline{v}(N) - \sum_{k \in N}v(k) = (n-1)\overline{v}(N) + \Delta \geq (n-1)\overline{v}(N).
	\]
	But this is a contradiction, because the opposite inequality holds. Now we proceed to prove that $e=v^i$.
	
	First, we show $i$ is the unique coalition of size $n-1$ with its coalition value $e(N\setminus i)$ different from $\overline{v}(N \setminus i)$, i.e.\, there is no $j \neq i$ such that $e(N \setminus j) < \overline{v}(N \setminus j)$. For a contradiction, if there is such $j$, denote $\varepsilon_i = \overline{v}(N \setminus i)-e(N \setminus i),\varepsilon_j=\overline{v}(N\setminus j) - e(N \setminus j)$ and $\varepsilon = \min\{\varepsilon_i,\varepsilon_j\}$. We define a non-trivial game $(N,x)$ such that both $(N,e+x) \in \oneCv$ and $(N,e-x) \in \oneCv$, contradicting (by Theorem~\ref{thm:extreme-points}) that $(N,e)$ is an extreme game. The game $(N,x)$ can be described as
	\[
	x(S) = \begin{cases}
		\varepsilon, & \text{if } S=N \setminus i \text{ or } S \notin \Kmin \wedge i \not\in S \wedge j \in S,\\
		-\varepsilon, & \text{if } S=N \setminus j \text{ or } S \notin \Kmin \wedge i \in S \wedge j \not\in S,\\
		0, & \text{otherwise.}
	\end{cases}
	\]
	Condition \eqref{def:1conv-cond1} from Definition~\ref{def:1conv-cond} for both $(N,e+x)$ now reads as
	\[
	\sum_{k \in N}e(N \setminus k) + x(N \setminus i) + x(N \setminus i) \leq (n-1)e(N) \]
	or equivalently
	\[ \sum_{k \in N}e(N \setminus k) + \varepsilon - \varepsilon \leq (n-1)e(N)
	\]
	and similarly for $(N,e-x)$ as $\sum_{k \in N}e(N \setminus k) - \varepsilon + \varepsilon \leq (n-1)e(N)$. Thus, for both game it is equivalent to the respective condition of $(N,e)$. Furthermore, for any nonempty coalition $S$ such that $i \not\in S$ and $j \in S$, Condition \eqref{def:1conv-cond2} from Definition~\ref{def:1conv-cond} for $(N,e+x)$ is
	\[
	e(S) + x(S) - (n-s-1)e(N) \leq\sum_{k \in N \setminus S}v(N \setminus k) + x(N \setminus i).
	\]
	Since $x(S)=x(N \setminus i)$, it is equivalent to the respective condition of $(N,e)$. Similarly, it holds for $(N,e-x)$ The rest of the cases for non-empty $S$ can be dealt with in a similar manner, therefore $(N,e+ x),(N,e-x) \in \oneC$. This leads to a contradiction, because $(N,e)$ is an extreme game.
	
	Second, we show that $e(N \setminus i)=v(N)-v(i)- \Delta$. Clearly, for $\alpha > 0$,
	\[
	(n-1)e(N)=\sum_{k \in N} e(N \setminus k) = nv(N) - \sum_{k \in N}v(k) - \alpha.
	\]
	We note that the first equality holds, otherwise there is a game $(N,x)$ such that $x(N \setminus j)\coloneqq\beta$ and for $S, i \not \in S: x(S)\coloneqq-\beta$. That leads to a contradiction with $(N,e)$ being an extreme game. Therefore, $\alpha = v(N) - \sum_{k \in N}v(k) = \Delta$.
	
	Finally, it is elementary to prove that $e(S)=e(N)-b(N \setminus S)=v^i(S)$. If this was not true, yet another game $(N,x)$ with $x(S)= e(N)-b(N \setminus S)$ and $x(T)=0$ would lead to a contradiction with extremality of $(N,e)$.
\qed\end{proof}

Now let us proceed with the investigation of the extreme rays.
For the game $(N,v^i + \lambda e)$ to be 1-convex (thus being in the recession cone of $\oneCv$), the following conditions must hold for every nonempty $S \subseteq N$:
\begin{itemize}
	\item $b^{v^i}(N) + b^{\lambda e}(N) \geq v^i(N) + \lambda e(N)$,
	\item $v^i(S) + \lambda e(S) \leq v^i(N) + \lambda e(N) - b^{v^i}(N \setminus S) - b^{\lambda e}(N \setminus S)$.
\end{itemize}
By Theorem~\ref{thm:excess-extending}, we have that
\begin{itemize}
	\item $b^{v^i}(N) = v^i(N)$ ,
	\item $v^i(S) = v^i(N) - b^{v^i}(N \setminus S)$.
\end{itemize}
We can therefore simplify the conditions, arriving at
\begin{itemize}
	\item $b^{\lambda e}(N) \geq \lambda e(N)$,
	\item $\lambda e(S) \leq \lambda e(N) - b^{\lambda e}(N \setminus S)$.
\end{itemize}
Furthermore, we can factor out $\lambda$ since it is non-negative. Notice that for each $j \in N$, $e(j) = e(N) = 0$, otherwise $(N,v^i+\lambda e) \not\in \oneCv$. Taking all this into consideration, we obtain the following conditions for $(N,e)$, representing an unbounded direction in $\oneC$:
\begin{enumerate}
	\item $b^e(N) \geq e(N) \iff \sum_{j \in N}e(N \setminus j) \leq 0$,		
	\item $\forall S \subseteq N, S \neq \emptyset: $ $e(S) \leq e(N) - b^e(N \setminus S) \iff e(S) \leq \sum_{j \in N \setminus S}e(N \setminus j)$,
	\item $\forall k \in N: 0 \leq \sum_{j \in N \setminus k}e(N \setminus k)$,
	\item $\forall j \in N:e(j) = 0$,
	\item $e(N)=0$.
\end{enumerate}
Conditions 1 and 2 show that $(N,e)$ itself has to be a 1-convex game. Moreover, if it is 1-convex, for any $\lambda \geq 0$, the game $(N,\lambda e)$ is also 1-convex. Therefore, the game $(N,e)$ is (not necessarily an extreme) ray of the recession cone of the set of $\oneC$-extensions. It is a zero-normalised game with $e(N)=0$ (thanks to conditions 4 and 5). Observe that condition 3 is a special case of condition 2 (take $S = \{k\}$ for $k \in N$). We state it separately, since it simplifies our further analysis. Notice an interesting fact: values of the game $(N,e)$ do not depend on the value of $(N,\Kmin,v)$. Therefore, the recession cone is the same for every minimal incomplete game.

Further, to simplify conditions $1$ to $5$, suppose that there is a $\oneC$-extension $(N,v^i + e)$ such that $\sum_{j \in N} e(N \setminus j) < 0$. Then there is $k \in N$ such that $\sum_{j \in N \setminus k}e(N \setminus j) < 0$. But this is a contradiction with $0 \leq \sum_{j \in N \setminus k}e(N \setminus j)$. Further, suppose that $\sum_{j \in N}e(N \setminus j) = 0$ and there is $k$ such that $e(N \setminus k) \neq 0$. If $e(N \setminus k) > 0$, then $e(N \setminus k) = - \sum_{j \in N \setminus k}e(N \setminus j)  > 0$, which is a contradiction because both $0 \leq \sum_{j \in N \setminus k}e(N \setminus j)$ and $\sum_{j \in N \setminus k}e(N \setminus j) < 0$. If $e(N \setminus k) < 0$, there is $\ell \in N$ such that $e(N \setminus \ell) > 0$ and we arrive into a similar contradiction. Hence, it must hold for every $i \in N$, that $e(N \setminus i) = 0$. 
We can now rewrite the conditions as
\begin{enumerate}
	\item $\forall S \subseteq N, S \neq \emptyset: e(S) \leq 0$,
	\item $\forall i \in N: e(i)=e(N \setminus i)=e(N)=0$.
\end{enumerate}
Let us now select the extreme rays. From Definition~\ref{thm:extreme-rays}, all but one of conditions $1$ or $2$ have to be satisfied with equality for the game $(N,v^i + e)$ to be an extreme ray. We see that the extreme rays are given by 1-convex games $(N,e_T)$ for coalitions $T \in E = 2^N \setminus \left(\{0,N\} \cup \{N\setminus i \mid i \in N \} \cup \{\{i\} \mid i \in N\}\right)$, where
\begin{equation}\label{eq:oneC-extreme-rays}
	e_T(S) \coloneqq \begin{cases}
		-1, & \text{if } S = T,\\
		0, & \text{if } S \neq T.\\
	\end{cases}
\end{equation}

With such knowledge, we are ready to fully describe the set of $\oneC$-extensions minimal incomplete games.

\begin{theorem}\label{thm:oneC-min-info-set}
	For a $\oneC$-extendable minimal incomplete game $(N,\Kmin,v)$, the set of $\oneC$-ex\-ten\-sions can be described as
	\[
	\oneCv = \left\{\sum_{i \in N}\alpha_i v^i + \sum_{ T \in E}\beta_Te_T \mid \sum_{i \in N}\alpha_i = 1 \text{ and } \alpha_i,\beta_T \geq 0\right\},
	\]
	where games $(N,v^i)$ are defined in~\eqref{eq:oneC-vertex-games}, games $(N,e_T)$ are defined in~\eqref{eq:oneC-extreme-rays} and $E= 2^N \setminus \left(\{0,N\} \cup \{N\setminus i \mid i \in N \} \cup \{\{i\} \mid i \in N\}\right)$.
\end{theorem}
\begin{proof}
	We have already proved that games $(N,v^i)$ for $i \in N$ from \eqref{eq:oneC-vertex-games} are the extreme games of $\oneCv$ and games $(N,e_T)$ for $T \in E$ from \eqref{eq:oneC-extreme-rays} are the extreme rays of $\oneCv$. The rest of the proof follows from Theorem~\ref{thm:char-of-polyhedrons}. 
\qed\end{proof}

\subsection{Values}\label{subsec:minimal-values}

We now introduce generalisations of the $\tau$-value and the Shapley value based on two ideas. The first idea is to consider solely the vertices of the set of $\oneC$-ex\-ten\-sions, compute their centre of gravity and compute its $\tau$-value (see Definition~\ref{def:average-tau}) or its Shapley value (Definition~\ref{def:average-shapley}) for the resulting game. The second idea considers also the recession cone, which is completely neglected in the first approach (Definitions~\ref{def:conic-tau}, \ref{def:conic-shapley}). We show that from the symmetry of recession cone, both approaches for both generalisations of the $\tau$-value and the Shapley value lead to the same solution concept for minimal incomplete games. We call it the \emph{average value}.

%We note that the definitions of values are defined for a more general class of incomplete games than we need in this section. This is because we shall use the definitions in Section~\ref{sec:upper} on incomplete games with defined upper vector.

%We remind the reader that we denote the minimal information with $\Kmin$.

\subsubsection{The average $\tau$-value}
The first solution concept considers the centre of gravity of the extreme games, that is \[\tilde{v}=\sum\limits_{i \in N} \frac{v^i}{N}.\]

Note that $(N,\tilde{v})$ is 1-convex if $(N,v^i)$ is 1-convex for every $i \in N$. Since additivity does not hold for the $\tau$-value in general, $\tau(\tilde{v}) \neq \sum_{i \in N} \frac{\tau(v^i)}{N}$ in general. We consider both variants\label{key} in the next definition.

\begin{definition}\label{def:average-tau}
	Let the games $(N,\tilde{v})$ and $(N,v^i)$ for $i \in N$ be the centre of gravity and the extreme games of $\oneCv$, respectively. 
	\begin{enumerate}
		\item The \emph{average $\tau$-value} $\tilde{\tau}\colon \oneC(\mathcal{K_{\min}}) \to \mathbb{R}^n$ is defined as
		\[
		\tilde{\tau}(v) \coloneqq  \tau(\tilde{v}).
		\]
		\item The \emph{solidarity $\tau$-value} $\tau^{s}\colon \oneC(\Kmin) \to \mathbb{R}^n$ is defined as
		\[
		\tau^{s}(v) \coloneqq  \sum\limits_{i \in N}\frac{\tau(v^i)}{n}.
		\]
	\end{enumerate}
\end{definition}

\begin{comment}
\begin{definition}\label{def:average-tau}
	Let $\K \subseteq 2^N$ and suppose that for every $v \in \oneC(\K)$, the set $\oneCv$ is a polyhedron described by its extreme points and extreme rays. Let the games $(N,\tilde{v})$ and $(N,v^i)$ for $i \in N$ be the centre of gravity and the extreme games of $\oneCv$, respectively. 
	\begin{enumerate}
		\item The \emph{average $\tau$-value} $\tilde{\tau}\colon \oneC(\mathcal{\K}) \to \mathbb{R}^n$ is defined as
		\[
		\tilde{\tau}(v) \coloneqq  \tau(\tilde{v}).
		\]
		\item The \emph{solidarity $\tau$-value} $\tau^{s}\colon \oneC(\K) \to \mathbb{R}^n$ is defined as
		\[
		\tau^{s}(v) \coloneqq  \sum\limits_{i \in N}\frac{\tau(v^i)}{n}.
		\]
	\end{enumerate}
\end{definition}
\end{comment}

The justification for the name of the solidarity $\tau$-value is given in the following theorem.

\begin{theorem}\label{thm:t-value-formula}
	The average $\tau$-value $\tilde{\tau}\colon\oneC(\Kmin)\to\Rn$ and the solidarity $\tau$-va\-lue $\tau^s\colon\oneC(\Kmin)\to\Rn$ can be expressed as follows:
	\begin{enumerate}
		\item $\forall i \in N: \tilde{\tau}_j(v)= v(j) + \frac{\Delta}{n}$,
		\item $\forall i \in N: \tau^{s}_j(v) = \frac{v(N)}{n}$.
	\end{enumerate}
\end{theorem}
\begin{proof}
	Both expressions can be easily derived from the definition of $\tilde{\tau}$, $\tau^{s}$ and extreme games $(N,v^i)$. First of all, the game $(N,\tilde{v})$ can be expressed as
	\[
	\tilde{v}(S)=\begin{cases}
		v(S), & \text{if } S \in \Kmin,\\
		v(N)-\left(\sum\limits_{j \in N \setminus S}v(j)\right) - \frac{n-s}{n}\Delta, & \text{if } S \not\in \Kmin.\\
	\end{cases}
	\]
	The values of its utopia vector are $b^{\tilde{v}}_j = v(j) + \frac{n-1}{n}\Delta$, and by summing them together over all players in $N$, we arrive at $b^{\tilde{v}}(N)=\sum_{j \in N}v(j) + (n-1)\Delta = v(N) + (n-2)\Delta$. The gap function for $N$ is $g^{\tilde{v}}(N) = (n-2)\Delta$ and finally (by Theorem~\ref{thm:1conv-formula}), from the definition of $\tilde{\tau}(v)$, using $\tilde{\tau}(v)=\tau(\tilde{v})$, we get
	\[
	\tau_j(\tilde{v}) = b^{\tilde{v}}_j - \frac{g^{\tilde{v}}(N)}{n} = v(j) + \frac{n-1}{n}\Delta - \frac{n-2}{n}\Delta = v(j) + \frac{\Delta}{n}.
	\]
	The main reason behind the formula for the solidarity $\tau$-value is the fact that $\tau(v^i)=b^{v^i}$. This immediately follows from $g^{v^i}(N) = 0$, and from the the form of $b^{v^i}$, which can be written as:
	\[
	b_j^{v^i} = \begin{cases}
		v(j), & \text{if }j=i,\\
		v(j) + \Delta, & \text{if }j\neq i.\\
	\end{cases}
	\]
	By summing the values of vector $b^{v^i}$, we get $b^{v^i}(N) = \sum_{j \in N}v(j) + \Delta = v(N)$. Therefore, $g^{v^i}(N) = v(N) - b^{v^i}(N)=0$. This implies the following equation:
	\[
	\tau^{s}_j = \sum_{i \in N}\frac{\tau_j(v^i)}{n} = \sum_{i \in N}\frac{b_j^{v^i}}{n} = \frac{\sum_{i \in N}b_j^{v^i}}{n},
	\]
	Finally, we can rewrite $\sum_{i \in N}b_j^{v^i}$ into:
	 \[\sum_{i \in N}b_j^{v^i} = \sum_{i \in N}v(i) + \Delta = v(N).\]
Using these equations, we arrive at the formula stated above.
\qed\end{proof}

We immediately see that the solidarity $\tau$-value is not a very reasonable solution concept if we consider that under such value, every player should get an equal share of $\frac{v(N)}{n}$ no matter his contribution.

\subsubsection{The conic $\tau$-value}
It might seem that the main downside of the previous two values is that we do not consider the recession cone of the set of $\oneC$-extensions. Here we provide an argument showing that with no further assumptions, this is not the case for minimal incomplete games. We define a solution concept dependent on the recession cone, which serves as a foundation for the study of incomplete games with more general sets $\K$.

Let $(N,v^i)$ for $i \in N$ be the extreme games of $\oneCv$ and $(N,e_T)$ for $T \in E$ be the extreme rays of $\oneCv$. $(N,\tilde{v})$ denotes again the centre of gravity of extreme games and $(N,\tilde{e})$ the centre of gravity of extreme rays, i.e.\ 
\[\tilde{e} = \sum\limits_{T \in E}\frac{e_T}{\lvert E \rvert} = \sum\limits_{T \in E}\frac{e_T}{\lvert 2^n - 2n- 2 \rvert}.\] 

The \emph{conic $\tau$-value} $\tau^<$, introduced in the following definition, is computed on the sum of these games. It considers the extreme games as well as extreme rays, thus the information from the shape of the conic cone is also included.
\begin{definition}
	The \emph{conic $\tau$-value} $\tau^<\colon \oneC(\Kmin)\to \mathbb{R}^n$ is defined as 
	\[
	\tau^<(v) \coloneqq \tau(\tilde{v} + \tilde{e}),
	\]
	where the games $(N,\tilde{v})$ and $(N,\tilde{e})$ are the centres of gravity of extreme points and of extreme rays of $\oneCv$.
\end{definition}

Surprisingly, the average $\tau$-value and conic $\tau$-value are the same function for minimal incomplete games. The reason is hidden in the symmetry of $(N,\tilde{e})$.

\begin{theorem}\label{thm:conic-equals-average}
	The conic $\tau$-value $\tau^{<}\colon \oneC(\Kmin)\to \Rn$ can be expressed as follows:
	\[
	\forall i \in N: \tau^<_i(v) = v(i) + \frac{\Delta}{n}.
	\]
\end{theorem}
\begin{proof}
	The proof is a straightforward derivation from the definitions. First, we already know the description of $(N,\tilde{v})$ from the proof of Theorem~\ref{thm:t-value-formula}. The description of $(N,\tilde{e})$ is
	\[
	\tilde{e}(S) = \begin{cases}
		-\frac{1}{\varepsilon}, & \text{if } S \in \Kmin \vee S= N\setminus j \text{ for } j \in N, \\
		0, & \text{otherwise},\\
	\end{cases}
	\]
	where $\varepsilon = 2^n - 2n-2$.
	From this description we derive that $b_i^{\tilde{v}+\tilde{e}} = v(i) + \frac{n-1}{n}\Delta - \frac{1}{\varepsilon}$ as $b_i^{\tilde{v}+\tilde{e}}$ is equal to
	\[
	(\tilde{v} + \tilde{e})(N) - (\tilde{v} + \tilde{e})(N\setminus j) = v(N) - \tilde{v}(N\setminus i) - \tilde{e}(N \setminus i)
	\]
	and the right-hand side can be rewritten as $v(N) - \left(v(N) - v(i) - \frac{n-1}{n}\Delta\right) + \frac{1}{\varepsilon}$. Further, $b^{(\tilde{v} + \tilde{e})}(N) = \sum\limits_{i \in N}v(i) + (n-1)\Delta + \frac{n}{\varepsilon} = v(N) + (n-2)\Delta + \frac{n}{\varepsilon}$. The last equality follows by using the fact that $\sum\limits_{i \in N} v(i) + \Delta = v(N)$. The gap function $g^{(\tilde{v} + \tilde{e})}(N) = b^{(\tilde{v} + \tilde{e})}(N) - v(N) = (n-2)\Delta + \frac{n}{\varepsilon}$.
	
	Finally, by Theorem~\ref{thm:1conv-formula}, $\tau_i^<(v) = b^{(\tilde{v} + \tilde{e})}_i - \frac{g^{(\tilde{v}+\tilde{e})}}{n} = v(i) + \frac{n-1}{n}\Delta + \frac{1}{\varepsilon} - \frac{(n-2)\Delta + \frac{n}{\varepsilon}}{n}$. Since $\frac{(n-2)\Delta + \frac{n}{\varepsilon}}{n} = \frac{n-2}{n}\Delta  + \frac{1}{\varepsilon}$, we have
	\[
	\tau^<_i(v) = v(i) + \frac{n-1}{n}\Delta + \frac{1}{\varepsilon} - \frac{n-2}{n}\Delta  - \frac{1}{\varepsilon} =  v(i) + \frac{\Delta}{n}.%\qedhere
	\] 
This completes the proof.\qed\end{proof}

As we already said, the reason for this (one might say surprising) result is the symmetry of $(N,\tilde{e})$. Actually, consider a more general setting in which take the expression
\begin{equation}\label{eq:gen-conic-game}
	\frac{1}{\gamma}\left(\beta\sum\limits_{i \in N}v^i + \alpha\sum\limits_{T \in E}e_T\right).
\end{equation}
This is a generalization of $\tilde{v} + \tilde{e}$ (as for $\beta = \frac{1}{n}, \alpha=\frac{1}{\varepsilon}$, and $\gamma = 1$, we get $\tilde{v} + \tilde{e}$). Also, if $\beta \neq \gamma$, it can be shown that the game from~(\ref{eq:gen-conic-game}) does not lie in $\oneCv$. Fixing $\beta=\gamma$, the $\tau$-value of this expression is equal to $\tilde{\tau}$ for any $\alpha \in \mathbb{R}$.

On the other hand, if $(N,\tilde{e})$ would not be symmetric or would depend on values of $v \in \oneC(\Kmin)$, the information about the cone might matter and it might be that $\tau^<(v) \neq \tilde{\tau}(v)$. This motivates the following definition for games $(N,\K,v)$ with a more general structure of $\K$.

\begin{definition}\label{def:conic-tau}
	Let $\K \subseteq 2^N$ and suppose that $\forall v \in \oneC(\K)$, the set $\oneCv$ is a polyhedron described by its extreme points and extreme rays. Then the \emph{$\alpha$-conic $\tau$-value} $\tau^\alpha \colon \oneC(\K)\to \mathbb{R}^n$ is defined as
	\[
	\tau^{\alpha}(v) \coloneqq  \tau(\tilde{v} + \alpha\tilde{e}),
	\]
	where $(N,\tilde{v})$, $(N,\tilde{e})$ are the centres of gravity of extreme points and of extreme rays of $\oneC(\K)$, respectively.
\end{definition}

As mentioned before, for minimal incomplete games, $\tilde{\tau}(v)=\tau^\alpha(v)$. However, for more general sets $\K$, this definition might yield a different solution. This is supported by the investigation of incomplete games with defined upper vector in Subsection~\ref{subsec:upper-values}. In there, we show that a similar solution concept, \emph{the $\alpha$-conic Shapley value}, does not coincide with \textit{the average Shapley value}. The idea behind these values is the same as behind the average and the conic $\tau$-value.

Once more, the fact that additivity does not hold for the $\tau$-value in general leads to a question whether $\tau(\tilde{v} + \tilde{e})$ and $\tau(\tilde{v}) + \tau(\tilde{e})$ yield different functions. For minimal incomplete games, this is not the case since $\tau(\tilde{e}) = 0$. Therefore, $\tau(\tilde{v}) + \tau(\tilde{e}) = \tau(\tilde{v}) = \tilde{\tau}(v)$.

\subsubsection{The average Shapley value}

The average Shapley value $\tilde{\phi}$ was already studied by Masuya and Inuiguchi in~\cite{Masuya2016} for the set of superadditive extensions of non-negative minimal incomplete games. We show that in the context of 1-convexity, the average Shapley value coincides with their definition, which is also equal to the average $\tau$-value. Yet again, the consideration of the recession cone (thanks to its symmetry) does not come to fruition.

\begin{definition}\label{def:average-shapley}
	The \emph{average Shapley value} $\tilde{\phi}\colon\oneC(\Kmin) \to \mathbb{R}^n$, is defined as
	\[
	\tilde{\phi}(v) \coloneqq \phi(\tilde{v}),
	\] 
	where $(N,\tilde{v})$ is the centre of gravity of extreme games of $\oneCv$.
\end{definition}

\begin{comment}
\begin{definition}\label{def:average-shapley}
	Let $\K \subseteq 2^N$ and suppose that for every $v \in \oneC(\K)$, the set $\oneCv$ is a polyhedron described by its extreme points and extreme rays.
	The \emph{average Shapley value} $\tilde{\phi}\colon\oneC(\K) \to \mathbb{R}^n$, is defined as
	\[
	\tilde{\phi}(v) \coloneqq \phi(\tilde{v}),
	\] 
	where $(N,\tilde{v})$ is the centre of gravity of extreme games of $\oneCv$.
\end{definition}
\end{comment}

For minimal incomplete games, the average Shapley value has the following simple formula.

\begin{theorem}\label{thm:ave-shapley-formula}
	The average Shapley value $\tilde{\phi}\colon\oneC(\Kmin)\to\Rn$ can be expressed as follows:
	\[
	\tilde{\phi}_i(v) = v(i)+\frac{\Delta}{n},\forall i \in N.
	\]
\end{theorem}
\begin{proof}
	The proof is based on the characterisation of the Shapley value from Theorem~\ref{thm:alternate-shapley-formula} and the fact that for every $S \subseteq N \setminus i$, $\tilde{v}(S \cup i) - \tilde{v}(S) = v(i) + \frac{\Delta}{n}$. This holds as $\tilde{v}(S \cup i) - \tilde{v}(S)$ is from the definition of $(N,\tilde{v})$ equal to
	\[
	v(N) - \sum\limits_{j \in N \setminus (S\cup i)}v(j) - \frac{(n-(s+1))}{n}\Delta - \left(v(N) - \sum\limits_{j \in N \setminus S}v(j)-\frac{(n-s)}{n}\Delta\right),
	\]
	which can be rewritten to $v(i) + \frac{\Delta}{n}$. Observe that $v(i) + \frac{\Delta}{n}$ is independent of coalition $S$. We know that $\tilde{\phi}_i(v) = \phi_i(\tilde{v})$ and substituting into the expression from Theorem~\ref{thm:alternate-shapley-formula}, we get
	\[
	\phi_i(\tilde{v}) = \frac{1}{n}\sum\limits_{S\subseteq N \setminus i}{n-1\choose s}^{-1}\left(v(i) + \frac{\Delta}{n}\right) = \left(v(i) + \frac{\Delta}{n}\right)\frac{1}{n}\sum\limits_{S\subseteq N \setminus i}{n-1\choose s}^{-1}.
	\]
	Modifying the sum is an easy exercise using the following identity:
	\[
	\sum\limits_{S \subseteq N \setminus i}{n-1 \choose{s}}^{-1}= \sum\limits_{j=0}^{n-1}{n-1\choose{j}}{n-1 \choose{j}}^{-1} = n.
	\]
	Combining together, we arrive at the desired formula.	
\qed\end{proof}

Similarly to the investigation of the conic $\tau$-value, we conclude that any sensible integration of the recession cone in the definition of the generalised Shapley value does not yield a different result. This is since $\phi(\tilde{v} + \alpha \tilde{e}) = \phi(\tilde{v}) + \phi(\alpha\tilde{e}) = \phi(\tilde{v})$ as for symmetric game $(N,\alpha\tilde{e})$, the Shapley value for any player $i$ equals $\phi_i(\alpha\tilde{e}) = 0$. Nonetheless, a similar argument for the definition of $\phi^\alpha$ for games with general $\K$ holds. For the conic Shapley value of incomplete games with defined upper vector, we show in Subsection~\ref{subsec:upper-values} that the two concepts do not coincide in general.

\begin{definition}\label{def:conic-shapley}
	Let $\K \subseteq 2^N$ and suppose that $\forall v \in \oneC(\K)$, the set $\oneCv$ is a polyhedron described by its extreme points and extreme rays. Then the \emph{$\alpha$-conic Shapley-value} $\phi^\alpha \colon \oneC(\K)\to \mathbb{R}^n$ is defined as
	\[
	\phi^{\alpha}(v) \coloneqq  \phi(\tilde{v} + \alpha\tilde{e}),
	\]
	where $(N,\tilde{v}),(N,\tilde{e})$ are the centres of gravity of extreme points and of extreme rays, respectively.
\end{definition}

To summarise, we considered generalisation of three values $\tau, n, \phi$ of complete cooperative games in two variants (including/excluding the information from the recession cone of $\oneC$-extensions) and showed that actually all of them coincide thanks to 1-convexity and symmetry of the recession cone of $\oneCv$. From now on, we refer to this value of minimal incomplete games as the \emph{average value} $\tilde{\zeta}$.

Let us now give a numerical example of computing such a value and then an example motivating the concept.

\begin{example}
	Let $(N,\Kmin,v)$ be a 5-player minimal incomplete game given by $v(1)=0$, $v(2)=v(4)=1$, $v(3)=2$, $v(5)=4$ and $v(N)=10$. 
	The total excess $\Delta$ is equal to 2. Thus the game is $\oneC$-extendable and the average value is \[\tilde{\zeta} = \left(\frac{2}{5},\frac{7}{5},\frac{12}{5},\frac{7}{5},\frac{22}{5}\right).\]
	Notice that even though the worth of singleton coalition ${1}$ is $0$, the payoff of player $1$ under the average value is $\frac{2}{5}$. This reflects the fact that without further assumptions, we cannot be sure that the total excess (or a proportion of it) is not generated by a participation of the player in the grand coalition.
\end{example}

\begin{example}
	Suppose we have a company employing prospective employees forming the set $N = \{1,\ldots,n\}$, each valued with salary $v(i)$. At the end of the year, we have been assigned a total budget $v(N)$ including bonuses. Thus the total budget is greater than the sum of the salaries (i.e. the total excess is $\Delta > 0$). Each employee is very ambitious and considers himself the one worthy of all the bonus money. Every one of them cannot be right (the world is not utopian). However, it makes sense to try to find a compromise. This leads to considering all the situations which are 1-convex. The average value then can be seen as the sought compromise since it considers all the possible 1-convex extensions and valuations of players under these extensions. In the result we assign each player the original salary plus the average benefit he brings into the company.
\end{example}

\subsection{Axiomatisation of the average value \texorpdfstring{$\tilde{\zeta}$}{}}\label{subsec:oneC-ave-value-axiom}\label{subsec:minimal-av}

In this subsection, we focus on axiomasations of the average value. We first consider known characterisations of the $\tau$-value and the Shapley value of complete games. We show how to generalise these characterisation for the average value. This is done with taking into account the fact that the average value is both the $\tau$-value and Shapley value of a specific complete game $(N,\tilde{v})$. In the second part, we offer three axiomatisations where the axioms are defined in the context of values of $v \in \oneC(\Kmin)$.

\subsubsection{Generalisations of known axiomatisations}

The idea behind these generalisations is based on the fact that the average value is defined as either the $\tau$-value or the Shapley value of the centre of gravity $(N,\tilde{v})$. Since these values satisfy certain axioms, also the average value satisfies these axioms when restricted to $\tilde{v}$. The uniqueness of the average value is then given by the uniqueness of $\tilde{\zeta}(v)$ for each $v \in \oneC(\Kmin)$. If we had a function $f$ satisfying the restricted axioms different from $\tilde{\zeta}$, we would have a game $v \in \oneC(\Kmin)$ such that $\tilde{\zeta}(v) = \tau(\tilde{v}) \neq f(v)$. But this means that for $(N,\tilde{v})$, we have two values for complete games satisfying the axioms of the $\tau$-value (or the Shapley value) that differ in the imputation assigned to $(N,\tilde{v})$. This is a contradiction with the uniqueness of these values.

We demonstrate this proof method on two examples, generalising both an axiomatisation of the $\tau$-value and the Shapley value.

\begin{theorem}\label{thm:axiom-known-tau}
	The average value $\tilde{\zeta}$ is the only function $f\colon \oneC(\Kmin)\to\mathbb{R}^n$ such that the following properties hold for every $v \in \oneC(\Kmin)$:
	\begin{enumerate}
		\item (efficiency) $\sum_{i \in N} f_i(v)=v(N)$,
		\item (restricted proportionality property of $\tilde{v}$) $f(v_0) = \alpha b^{\tilde{v}_0}$, where $v_0$ is the zero-normalised game of $v$, and
		\item (minimal right property of $\tilde{v}$) $f(v) = a^{\tilde{v}} + f(v-a^{\tilde{v}})$, where $\alpha \in \mathbb{R}$ and $(v-a^{\tilde{v}})(S) \coloneqq v(S) - \sum_{i \in S}a^{\tilde{v}}_i$ for every $S \subseteq N$.
	\end{enumerate}
\end{theorem}
\begin{proof}
	To prove that the average value satisfies the mentioned properties, recall the definition $\tilde{\zeta}(v)=\tau(\tilde{v})$ and Theorem~\ref{thm:tau-axiom1}. Since $\tilde{v}(S)=v(S)$ for $S \in \Kmin$, all three properties hold.
	
	Regarding uniqueness, suppose there is a function $g\colon \oneC(\Kmin)\to\mathbb{R}^n$ such that the properties hold and there is a game $v \in \oneC(\Kmin)$, $\tilde{\zeta}(v)\neq g(v)$. We can construct a function $\gamma\colon \oneC\to\mathbb{R}^n$ such that $\gamma(w)\coloneqq \tau(w)$ for every $w \in \oneC$, $w \neq \tilde{v}$ and $\gamma(\tilde{v})\coloneqq g(v)$. Clearly, $\gamma$ satisfies all axioms from Theorem~\ref{thm:tau-axiom1}, which leads (together with $\gamma(\tilde{v})=g(v)\neq\tau(\tilde{v})$) to a contradiction with the uniqueness of the $\tau$-value.
\qed\end{proof}

It can be shown that in the context of incomplete games the second axiom is equivalent to restricted proportionality property of $\tilde{v}$ where $\alpha=1$.

The alternative characterisation of the $\tau$-value was proposed in~\cite{Tijs1995} and it can be generalised in a similar manner. Let us proceed with yet another example, generalising axiomatisation of the Shapley value.

\begin{theorem}\label{thm:axiom-known-shapley}
	The average value $\tilde{\zeta}$ is the only function $f\colon \oneC(\Kmin)\to\mathbb{R}^n$ such that the following properties hold for every $v,w \in \oneC(\Kmin)$:
	\begin{enumerate}
		\item (\textit{efficiency}) $\sum_{i\in N}f_i(v) = v(N)$,
		\item (\textit{symmetry of $\tilde{v}$})  $\forall i,j \in N, \forall S \subseteq N \setminus \{i,j\}: v(S\cup i) = v(S \cup j) \implies f_i(v)=f_j(v)$,
		\item (\textit{null player of $\tilde{v}$}) $\forall i \in N, \forall S \subseteq N\setminus i:\tilde{v}(S) = \tilde{v}(S+i) \implies f_i(v)=0$,
		\item (\textit{additivity of $\tilde{v}$}) $f(\tilde{v}+\tilde{w})=f(\tilde{v}) + f(\tilde{w})$.
	\end{enumerate}
\end{theorem}
\begin{proof}
	Since the average value of $v \in \oneC(\Kmin)$ acts as the Shapley value of $\tilde{v} \in \oneC$, the axioms are satisfied. We note that considering efficiency, $\sum_{i \in N}f_i(v) = v(N) = \tilde{v}(N)$, therefore it is equivalent with $\phi(\tilde{v})=\tilde{v}(N)$ and for additivity, $(\tilde{v}+\tilde{w})=\tilde{v} + \tilde{w}$, where $(\tilde{v}+\tilde{w})$ is the centre of gravity of vertices of $v+w$ and $\tilde{v}+\tilde{w}$ is the sum of centres of gravity of $v$ and $w$. The uniqueness of $\tilde{\zeta}$ is given by the uniqueness of the Shapley value.
\qed\end{proof}

From the alternative characterisations of the Shapley value, we generalised those in~\cite{Brink2002,Young1989}.
To do the same for the one by Roth~\cite{Roth1977} seems to be more challenging.

\subsubsection{Axiomatisations in the context of values of $(N,\Kmin,v)$}

The previously mentioned characterisations do not tell us anything new about the average value that we do not already know from its definition $\tilde{\zeta}(v)\coloneqq\tau(\tilde{v})=\phi(\tilde{v})$. In~this subsection, we derive three axiomatisations in the context of values of $(N,\Kmin,v)$.

\begin{theorem}\label{thm:ave-value-axiom1}
	The average value $\tilde{\zeta}$ is the only function $f\colon \oneC(\Kmin)\to\mathbb{R}^n$ such that the following properties hold for every $v \in \oneC(\Kmin)$:
	\begin{enumerate}
		\item (efficiency) $\sum_{i \in N}f_i(v) = v(N)$,
		\item (elementary symmetry) $\forall i,j \in N: i \neq j \wedge v(i)=v(j) \implies f_i(v)=f_j(v)$,
		\item (zero-normalisation invariance) $\forall i \in N: f_i(v) = v(i) + f_i(v_0)$, where $v_0$ is the zero-normalised game of $v$.
	\end{enumerate}
\end{theorem}
\begin{proof}
	Let us prove that $\tilde{\zeta}$ satisfies all three properties. First, 	
	\[
	\sum_{i \in N}\tilde{\zeta}_i(v) = \sum_{i \in N}v(i) + n \frac{\Delta}{n} = v(N).
	\]
	Furthermore, for $v(i)=v(j)$, it holds $\tilde{\zeta}_i(v) = v(i) + \frac{\Delta}{n} = v(j) + \frac{\Delta}{n} = \tilde{\zeta}_j(v)$. For the third property, as  $v(i) + \frac{\Delta}{n} = \tilde{\zeta}_i(v) =  v(i) + \tilde{\zeta}_i(v_0)$ holds for every player $i$, it suffices to show that $\tilde{\zeta}_i(v_0) = \frac{\Delta}{n}$.

	The total excess $\Delta_0$ of $(N,\Kmin,v_0)$ is equal to $v(N)-\sum_{i \in N}v(i)=\Delta$. Therefore $\tilde{\zeta}_i(v_0) = v_0(i) + \frac{\Delta_0}{n} = \frac{\Delta}{n}$, thus the third property also holds.
	
	Now, let us prove that $f\colon \oneC(\Kmin)\to\mathbb{R}^n$ satisfying the three properties must be $\tilde{\zeta}$. First, from zero-normalisation invariance, it holds $\forall i: f_i(v) = v(i) + f_i(v_0)$. Because $\tilde{\zeta}_i(v) = v(i) + \frac{\Delta}{n}$ for all $i \in N$, it suffices to prove that $f_i(v_0) = \frac{\Delta}{n}$ for any zero-normalised game $v_0$ of $v \in \oneC(\Kmin)$. From the first property, we have $\sum_{i \in N}f_i(v_0) = \Delta = v_0(N)$. Also, $v_0(i)=v_0(j)$ for all pairs of players $i,j$ implies $f_i(v_0) = f_j(v_0)$. Combining both together, we get $f_i(v_0) = \frac{\sum_{j \in N} f_i(v_0)}{n} = \frac{\Delta}{n}$.
\qed\end{proof}

Another characterisation can be obtained by substituting the axiom of zero-normalisation invariance for additivity axiom. Such replacement has to be compensated by adding yet another axiom, because without it, they do not characterise the function uniquely (for example, the solidarity $\tau$-value $\tau^{s}$ also satisfies these three axioms). We deal with this by providing two different axioms: \emph{zero-excess} axiom (employing the total excess $\Delta$) and a more familiar axiom of \emph{individual rationality}.

\begin{theorem}\label{thm:ave-value-axiom2}
	The average value $\tilde{\zeta}$ is the only function $f\colon \oneC(\Kmin)\to\mathbb{R}^n$ such that the following properties hold for every $v,w \in \oneC(\Kmin)$:
	\begin{enumerate}
		\item (efficiency) $\sum_{i \in N}f_i(v) = v(N)$,
		\item (elementary symmetry) $\forall i,j \in N: i \neq j \wedge v(i)=v(j) \implies f_i(v)=f_j(v)$,
		\item (elementary additivity) $f(v + w) = f(v) + f(w)$,
		\item (zero-excess axiom) if $\Delta_v=0 \implies\forall i \in N:  f_i(v)=v(i)$,
		\item (individual rationality) $\forall i \in N: f_i(v)\geq v(i)$.
	\end{enumerate}
\end{theorem}
\begin{proof}
	We have already proved in Theorem~\ref{thm:ave-value-axiom1} that the first two axioms are satisfied by $\tilde{\zeta}$. To prove additivity, we have $\tilde{\zeta}_i(v+w)=v(i)+w(i) + \frac{\Delta_{v+w}}{n}$ and
	\[
	\tilde{\zeta}_i(v) + \tilde{\zeta}_i(w) = v(i) + \frac{\Delta_v}{n} + w(i) + \frac{\Delta_w}{n} = v(i) + w(i) + \frac{\Delta_v}{n} + \frac{\Delta_w}{n}.
	\]
	for any player $i$. Clearly, if $\Delta_{v+w} = \Delta_v + \Delta_w$, elementary additivity is satisfied. However,
	\[
	\Delta_{v+w}=v(N) + w(N) - \sum_{i \in N}\left(v(i) + w(i)\right)=
	v(N) - \sum_{i \in N}v(i) + w(N) - \sum_{i \in N}w(i)
	\]
	and
	\[
	\Delta_v + \Delta_w = v(N) - \sum_{i \in N}v(i) + w(N) - \sum_{i \in N}w(i).
	\]
	Zero-excess axiom is satisfied because for $\Delta_v=0$ and any player $i$, $\tilde{\zeta}_i(v)=v(i) + \frac{\Delta_v}{n} = v(i)$. Individual rationality is also satisfied as for any player $i$: $\tilde{\zeta}_i(v) = v(i) + \frac{\Delta_v}{n} \geq v(i)$.
	
	To substitute elementary additivity for zero-normalisation invariance in our proof of uniqueness, we define a minimal incomplete game $(N,\Kmin,\Sigma)$ such that $v=v_0 + \Sigma$. We do so by setting $\Sigma(i)\coloneqq v(i)$ and $\Sigma(N) \coloneqq \sum_{i \in N}v(i)$. Notice that $\Delta_\Sigma = 0$ and thus, $\Sigma \in \oneC(\Kmin)$. Now, from elementary additivity $f_i(v)=f_i(v_0) + f_i(\Sigma)$. We already proved, that $f_i(v_0)=\frac{\Delta}{n}$, thus all that remains is to prove that for every player $i$, $f_i(\Sigma)=v(i)$.
	
	From zero-excess axiom, this already holds as $\Delta_\Sigma=0$. Without zero-excess axiom, by efficiency $\sum_{i \in N}f_i(\Sigma)=\Sigma(N)=0$ and individual rationality, $f_i(\Sigma) \geq 0$, which leads together to the desired $f_i(\Sigma)=v(i)$.
\qed\end{proof}

We conclude this section with further axioms, which are connected with different definitions of the Shapley value~\cite{Brink2002,Young1989}.

Both of the following properties can be easily derived from the definition of the average value $\tilde{\zeta}$.

\begin{theorem} For the average value $\tilde{\zeta}\colon\oneC(\Kmin)\to \Rn$, the following properties hold for every $v,w \in \oneC(\Kmin)$:
	\item (elementary triviality) $\tilde{\zeta}(v_\emptyset)=0$, where $v_\emptyset(S)\coloneqq0 \text{ for } S \in \Kmin$,
	\item (elementary fairness) $\tilde{\zeta}_i(v+w) - \tilde{\zeta}_i(v) = \tilde{\zeta}_j(v+w) - \tilde{\zeta}_j(v)$ if $w(i)=w(j)$. \qed
\end{theorem}

Notice that a property similar to \textit{null player} cannot hold when $\Delta >0$. That is because if $v(i)=0$ for a player $i$, then $\tilde{\zeta}_i = \frac{\Delta}{n} > 0$. This might seem surprising since in the characterisation of the Shapley value, the axiom of null player is satisfied. This corresponds with the idea that even though it might seem from the known information given by $\Kmin$ that the player does not have any worth in the game, since we cannot be sure, we act as if he has some.

\section{Incomplete games with defined upper vector}\label{sec:upper}

By an incomplete game with defined upper vector, we mean a game $(N,\K,v)$ such that $\{N\}\cup\{N\setminus i\mid i \in N\} \subseteq \K$. Similarly to the previous section, we first derive a description of the set of $\oneC$-extensions and then we focus on values. Our results show that the average Shapley value and the conic Shapley value does not coincide in general for incomplete games with defined upper vector with at least four players.

\subsection{Description of the set of $\oneC$-extensions}\label{subsec:upper-description}
Similarly to the previous section, we initiate the derivation of a description of $\oneCv$ by a characterisation of $\oneC$-extendability.
\begin{theorem}\label{thm:duv-extendability}
	Let $(N,\K,v)$ be an incomplete game with defined upper vector. It is $\oneC$-extendable if and only if
	\begin{equation}\label{eq:upper1}
		\forall S \in K: v(S) \leq v(N) - b(N\setminus S)
	\end{equation}
	and
	\begin{equation}\label{eq:upper2}
		b(N) \geq v(N).
	\end{equation}
\end{theorem}
\begin{proof}
	If the conditions hold, we can define a complete game $(N,\overline{v})$ such that 
	\begin{equation}\label{eq:oneC-UV-upper-game}
	\overline{v}(S) = \begin{cases}
		v(S), & \text{if } S \in \K,\\
		v(N)-b(N \setminus S), & \text{if } S \not\in \K.\\
	\end{cases}
	\end{equation}
	The game is 1-convex, because $v(S) \leq v(N) - b(N \setminus S)$ for $S \not\in \K$ holds since the left-hand side is actually equal to the right-hand side. For $S \in \K$ the conditions hold from the assumption as well as the condition $b(N) \geq v(N)$. Therefore, it is a $\oneC$-extension of $(N,\K,v)$
	
	If one of the conditions~\eqref{eq:upper1} or~\eqref{eq:upper2} fails, the condition does not hold for any extension, therefore the extension is not 1-convex.
\qed\end{proof}

We denote the $\oneC$-extension using the line over $v$. This is not a coincidence, as the game is really the upper game of the set of $\oneCv$-extensions. On the top of that, it is also the only extreme game of the set.

\begin{theorem}\label{thm:duv-extreme}
	Let $(N,v)$ be a $C_1^n$-extendable incomplete game with defined upper vector. Then game $(N,\overline{v})$, defined in~\eqref{eq:oneC-UV-upper-game}, is the only extreme game of $C_1^n(v)$.
\end{theorem}
\begin{proof}
	First, let us prove it is an extreme game. Following Theorem~\ref{thm:extreme-points}, let us consider $(N,x)$ such that both $\overline{v}+x$ and $\overline{v}-x$ are $C_1^n(v)$-extensions. If for any $S$, the value $x(S)>0$, then either $\overline{v}+x$ is not a $\oneC$-extension (if $S \in \K$, then $v(S)\neq(\overline{v}+x)(S)$) or the complete game is not 1-convex (if $S \not\in \K$, then $(\overline{v}+x)(S)=\overline{v}(S) + x(S) = v(N) - b(N \setminus S) + x(S) > v(N) - b(N \setminus S)$.
	
	Further, suppose there is an extreme game $(N,y)$ different from $(N,\overline{v})$. It means there is a coalition $S \not\in \K$ such that $y(S) < \overline{v}(S)$. If we take $(N,x)$ such that $x(S) = \overline{v}(S)-y(S)$ and $x(T)=0$ otherwise, we immediately conclude that both $y+x$ and $y-x$ are in $C_1^n(v)$, and since $x\neq 0$, we conclude by Theorem~\ref{thm:extreme-points} that $(N,y)$ is not an extreme game.
\qed\end{proof}

We further define games $(N,e_T)$ for $T \not\in \K$ as
\begin{equation}\label{eq:oneC-UV-extreme-ray}
e_T(S) \coloneqq \begin{cases}
	-1, & \text{if } S=T,\\
	0, & \text{otherwise.}\\
\end{cases}
\end{equation}
It is not hard to see that games $(N,\overline{v}+\alpha e_T)$ are $\oneCv$-extensions for any $\alpha\geq 0$, therefore $(N,e_T)$ are rays of $\oneCv$. Moreover, all but one conditions are satisfied for $(N,\overline{v}+e_T)$ with equality, therefore they are even the extreme rays. The following theorem shows they are the only extreme rays.

\begin{theorem}\label{thm:duv-description}
	Let $(N,v)$ be a $\oneC$-extendable incomplete game with defined upper vector. Then the set of $\oneCv$-extensions can be described as
	\[
	C_1^n(v)=\left\{\overline{v} + \sum_{T \not\in\K}\alpha_Te_T \mid \alpha_T \geq 0\right\},
	\]
	where $(N,\overline{v})$ is defined in~\eqref{eq:oneC-UV-upper-game} and games $(N,e_T)$ are defined in~\eqref{eq:oneC-UV-extreme-ray}.
\end{theorem}
\begin{proof}
	For a $\oneC$-extension $(N,w)$, we show it can we expressed as a combination of the upper game and games $(N,e_T)$ for $T \not\in \K$. Since $(N,w) \in C_1^n(v)$, it holds for every $T \not\in \K$ that $w(T) \leq \overline{v}(T)$. Therefore, we define $\alpha_T \coloneqq  w(T) - \overline{v}(T)$. Immediately, it follows that 
	\[(w + \alpha_Te_T)(T) = w(T) -w(T) + \overline{v}(T) = \overline{v}(T).\]
	Setting $\alpha_T$ for every $T \not\in \K$ in this manner concludes the proof. 
\qed\end{proof}

\subsection{Values}\label{subsec:upper-values}

Definitions of the average $\tau$-value (Definition~\ref{def:average-tau}) and of the average Shapley value (Definition~\ref{def:average-shapley}) can be readily extended to the incomplete games with defined upper vector. This can be done similarly to Definitions~\ref{def:conic-tau} and~\ref{def:conic-shapley} because of the structure of the set of $\oneC$-extensions of such games.

% The definitions of the average $\tau$-value (Definition~\ref{def:average-tau}), the average Shapley value (Definition~\ref{def:average-shapley}) and the definitions of their conic variants (Definitions~\ref{def:conic-tau} and~\ref{def:conic-shapley}) apply also to the incomplete games with defined upper vector because of the structure of the set of $\oneC$-extensions of such games.

From the point of view of the $\tau$-value, a game with defined upper vector is equivalent with any of its $\oneC$-extensions. This is because, for a complete game $v$, $\tau(v)=b^v-\frac{g^v(N)}{n}$ and both $b^v$ and $g^v(N)$ depend only on values $v(N)$ and $v(N\setminus i)$ for $i \in N$.

Incomplete games with defined upper vector are a good example for showing that in general, $\tilde{\phi}(v)\neq\phi^\alpha(v)$. From the definition of the conic Shapley value, additivity and S-equivalence axiom, we have
\[
\phi^\alpha(v) = \phi(\tilde{v} + \alpha\tilde{e}) = \phi(\tilde{v}) + \alpha \phi(\tilde{e}).
\]
Therefore, $\phi^\alpha(v)=\tilde{\phi}(v)$ for $\alpha > 0$ if and only if $\alpha\phi(\tilde{e})=0$, which is equivalent to $\phi(\tilde{e})=0$. In order to compute the Shapley value, we need to obtain the marginal contributions of player $i$ to all coalitions $S$, i.e.\ $\tilde{e}(S\cup i) - \tilde{e}(S)$ for $i \in N$ and $S \subseteq N \setminus i$.

\begin{lemma}\label{lem:marginal-contributions}
	Let $(N,\K,v)$ be a $C_1^n$-extendable incomplete game with defined upper vector and for $C_1^n(v)$, let the game $(N,\tilde{e})$ be the centre of gravity of its extreme rays. Then we have
	\[
	\tilde{e}(S \cup i) - \tilde{e}(S) = \begin{cases}
		0, & \text{if } S \in \K\text{ and } S \cup i \in \K, \\
		0, & \text{if } S \not\in \K\text{ and } S \cup i \not\in \K,\\
		\frac{1}{\lvert E \rvert}, & \text{if } S \not\in \K\text{ and }S\cup i \in \K,\\
		-\frac{1}{\lvert E \rvert}, & \text{if } S \in \K\text{ and }S \cup i \not\in\K,\\
	\end{cases}
	\]
	where $E = \{T\subseteq N \mid T \in 2^N \setminus \K\}$ and thus $\lvert E \rvert = 2^n - \lvert \K \rvert$.
\end{lemma}
\begin{proof}
	We denote by $\overline{\K}$ the coalitions with unknown values, that is $\overline{\K}\coloneqq 2^N \setminus \K$. From the definition of $(N,\tilde{e})$, we have 
	\[
	\tilde{e}(S \cup i) - \tilde{e}(S) = \frac{1}{\lvert E \rvert}\left(\sum_{T \in \cK} e_T(S \cup i)-\sum_{T \in \cK}v(S)\right).
	\]
	Remember, that $e_T(S)=-1$ if and only if $T=S$, otherwise $e_T(S)=0$. It means that if $S \in \cK$, the sum $\sum_{T \in \cK}e_T(S)$ is equal to zero and similarly for $S\cup i$. Let us now distinguish the following cases.
	
	\begin{itemize}
		\item If $S \in \K$ and $S\cup i \in \K$, we have \[\left(\sum_{T \in \cK} e_T(S \cup i)-\sum_{T \in \cK}v(S)\right) = 0-0=0.\]
		\item If $S \in \cK$ and $S\cup i \in\cK$, we have \[\left(\sum_{T \in \cK} e_T(S \cup i)-\sum_{T \in \cK}v(S)\right) =-1 - (-1) = 0.\]
		\item If $S \in \K$ and $S\cup i \in\cK$, we have \[\left(\sum_{T \in \cK} e_T(S \cup i)-\sum_{T \in \cK}v(S)\right) = 0 - (-1) = 1.\]
		\item If $S \in \cK$ and $S\cup i \in \K$, then \[\left(\sum_{T \in \cK} e_T(S \cup i)-\sum_{T \in \cK}v(S)\right) = -1 - 0 = -1.\]
	\end{itemize}
	This case analysis concludes the proof.
\qed\end{proof}

\begin{lemma}\label{lem:Shapley-of-centre-of-rays}
	Let $(N,\K,v)$ be $C_1^n$-extendable incomplete game with defined upper vector and for $C_1^n(v)$, game $(N,\tilde{e})$ the centre of gravity of its extreme rays. Then it holds
	\[
	\phi_i(\tilde{e}) = \frac{1}{\lvert E \rvert n!}\left(\sum\limits_{\substack{S \subseteq N \setminus i \\ S \in \K \\ S \cup i \in \cK}}s!(n-s-1)! - \sum\limits_{\substack{S \subseteq N \setminus i \\ S \in \cK \\ S \cup i \in \K}}s!(n-s-1)!\right).
	\]
\end{lemma}
\begin{proof}
	The result immediately follows from the definition of the Shapley value (Definition~\ref{def:shapley}) and Lemma~\ref{lem:marginal-contributions}, since $\phi_i(\tilde{e}) = \frac{1}{\lvert E \rvert n!}\sum_{S \subseteq N \setminus i}(\dots)$ and substituting corresponding marginal contributions and dividing the sum into 4 sums according to presence of $S$ and $S \cup i$ in $\K$ yields the formula above.
\qed\end{proof}

From Lemma~\ref{lem:Shapley-of-centre-of-rays}, we can conclude that for cooperative games with at most $3$ players, $\tilde{\phi}$ and $\phi^\alpha$ always coincide. However, if we consider games with more players, the two values differ.

\begin{theorem}\label{thm:shapley-same}
	Let $(N,\K,v)$ be a $\oneC$-extendable incomplete game with defined upper vector.
	\begin{enumerate}
		\item If $n \leq 3$, then $\tilde{\phi}(v)=\phi^\alpha(v)$,
		\item If $n \geq 4$ and $\K = 2^N \setminus \{\{i\} \mid i \in N\}$, then $\tilde{\phi}(v)\neq\phi^\alpha(v)$.
	\end{enumerate}
\end{theorem}
\begin{proof}
	If $n \leq 3$, then if there is $S \not\in \K$, it is a singleton coalition $S=\{j\}$. This means, that in the case $S \not\in \K$ and $S \cup i \in \K$, the element of the sum is $s!(n-s-1)!=0$. Also, if we consider the other sum where $T \in \K$ and $T \cup i \not\in\K$, the only possibility is $T=\emptyset$, therefore, again $t!(n-t-1)!=0$. Thus $\phi_i(\tilde{e})=0$ for any such game and any player $i$, leading to coincidence of $\tilde{\phi}$ and $\phi^\alpha$.
	
	For $n \geq 4$ and $\K = 2^N \setminus \{i \mid i \in N\}$, the coalition $S$ satisfying $S \in \K$ and $S \cup i \not\in\K$ is only $S=\emptyset$. for $j\neq k$. Since there are ${N\choose 2}$ coalitions $\{j,k\}$ for every $j$ and there are $n$ players, we get that the second sum is equal to $n(n-1)1!(n-2)!=n!$. Further, the coalitions $S$ satisfying $S \not\in\K$ and $S \cup i \in \K$ are only those satisfying $\lvert S \rvert = n-2$. For every such coalition $s!(n-s-1)! = (n-2)!(n-(n-2)-1)!=(n-2)!$ and there are $2{N\choose 2} =2\frac{n!}{(n-2)!}2!=n!$. Therefore, $\tilde{\phi}_i(v)\neq\phi^\alpha(v)$.
\qed\end{proof}

\section{Conclusion}\label{sec:conclusion}

To the lack of our knowledge, this work is the first treatment of 1-convexity and $\oneC$-extensions in the framework of incomplete games. We focused on two prominent classes of games with incomplete information: minimal incomplete games and incomplete games with defined upper vector.

A natural direction for future research is to go beyond the case of minimal information and defined upper vector. We have already seen that describing the set of $\oneC$-extensions can be nontrivial even for these classes of incomplete games. Also, some generalised values may coincide or differ depending on the structure of $\K$. This is exemplified in our work in Section~\ref{sec:upper} on games with defined upper vector.

We believe that our results can be seen as a first larger step towards understanding other structures of $\K$ and questions related to 1-convexity. We believe that the techniques presented in this paper can be useful in the analysis of core catchers and the classes of quasi-balanced and semi-balanced games. Indeed, some of those problems are a work in progress.

\begin{acknowledgements}
The authors would like to thank Milan Hlad\'{i}k and David Hartman for initial discussions regarding the paper.
\end{acknowledgements}

\section*{Conflict of interest}

The authors declare that they have no conflict of interest.

\bibliographystyle{spmpsci}
\bibliography{bibliography}

\end{document}